\documentclass[11pt,reqno,a4paper]{amsart}
\usepackage[utf8]{inputenc}
\usepackage[english,spanish,activeacute]{babel}
\usepackage[T1]{fontenc}
\usepackage{amssymb,amsmath,amsthm,mathtools}
\usepackage{stmaryrd}
\usepackage{hyphenat}
\usepackage{pgfplots}
\usepackage{caption}
\usepackage{subcaption}
\usepackage{blkarray}
\usepackage{enumerate}
\usepackage{enumitem}

\usepackage[a4paper,top=3.3cm,bottom=3cm,left=3cm,right=3cm,bindingoffset=5mm]{geometry}

\usepackage{pdfpages}
\usepackage{verbatim}
\usepackage{tikz}
\usepackage{url}
\usepackage{color}
\usepackage{mathrsfs}
\usepackage[bookmarks,pdftex,backref=page,linktocpage,colorlinks=true]{hyperref}%backref=page,
\usepackage{cleveref}
\hypersetup{
	colorlinks=true,
	linkcolor=cyan,
	citecolor = brown,
	filecolor=magenta,
	urlcolor=teal,
}
\usepackage{float}%Para usar H
\usepackage{multicol}

\usepackage{xspace}
\usepackage{xfrac}
\usepackage{faktor}
\usepackage{MnSymbol}
\usepackage{relsize}
\usepackage{float}
\usepackage{multirow}
\usepackage{ifsym}
\usepackage[normalem]{ulem}
\usepackage{changepage}

\usepackage{amsfonts}
\usepackage{indentfirst}
\usepackage[mathscr]{eucal}
\usepackage{hhline,colortbl}
\usepackage{tabularx}
\usepackage{amscd}

\usepackage{latexsym,epigraph,bbm,upgreek}  %math symbols
\usepackage{tikz-cd} % commutative diagrams
\usepackage{array}
\usepackage{diagbox}

\theoremstyle{plain}
\newtheorem{thm}{Theorem}[section]

\newtheorem{prop}[thm]{Proposition}

\theoremstyle{definition}
\newtheorem{mydef}[thm]{Definition}

\newtheorem{remark}[thm]{Remark}

\newtheorem*{A}{Theorem A}

\DeclareMathOperator{\dis}{d}

\DeclareMathOperator{\ev}{ev}

\RequirePackage{etoolbox}
\DeclarePairedDelimiterX\conj[1]{\lbrace}{\rbrace}{#1} % for sets
\DeclarePairedDelimiterX\paren[1]{\lparen}{\rparen}{#1} % for parenthesis
\DeclarePairedDelimiterX\card[1]{\lvert}{\rvert}{#1} % cardinal of a set
\DeclarePairedDelimiterX\gen[1]{\langle}{\rangle}{\ifblank{#1}{\:\cdot\:}{#1}} % subspace generated by
\DeclarePairedDelimiterX\abs[1]{\lvert}{\rvert}{#1}% absolute value

\DeclarePairedDelimiter\floor{\lfloor}{\rfloor}

\DeclarePairedDelimiterX\braket[2]{\langle}{\rangle}{#1\,\delimsize\vert\,\mathopen{}#2}

\newcommand{\N}{\mathbb{N}}
\newcommand{\Z}{\mathbb{Z}}

\newcommand{\F}{\mathbb{F}}
\newcommand{\bs}{\boldsymbol}

\newcommand{\Co}{\mathcal{C}}
\newcommand{\Su}{\mathcal{S}}

\title[New quantum codes from homothetic-BCH codes]{New quantum codes from homothetic-BCH codes}

\author{C. Galindo}
\email[Carlos Galindo]{galindo@uji.es}
\address{Instituto Universitario de Matemáticas y Aplicaciones de Castellón and Departamento de
	Matemáticas, Universitat Jaume I, Campus de Riu Sec., 12071 Castelló, Spain}

\author{F. Hernando}
\email[Fernando Hernando]{carrillf@uji.es}
\address{Instituto Universitario de Matemáticas y Aplicaciones de Castellón and Departamento de
	Matemáticas, Universitat Jaume I, Campus de Riu Sec., 12071 Castelló, Spain}

\author{H. Martín-Cruz}
\email[Helena Martín-Cruz]{martinh@uji.es}
\address{Instituto Universitario de Matemáticas y Aplicaciones de Castellón and Departamento de
	Matemáticas, Universitat Jaume I, Campus de Riu Sec., 12071 Castelló, Spain}

\thanks{This work has been partially supported by
	MCIN/AEI/10.13039/501100011033 and by the “European Union NextGenerationEU/PRTR”, grant PID2022-138906NB-C22, as well as by Universitat Jaume I, grants GACUJIMA/2024/03 and PREDOC/2020/39.}

\begin{document}
	
\selectlanguage{english}
	
\begin{abstract}
We introduce homothetic-BCH codes. These are a family of $q^2$-ary classical codes $\mathcal{C}$ of length $\lambda n_1$, where $\lambda$ and $n_1$ are suitable positive integers such that the punctured code $\mathcal{B}$ of $\mathcal{C}$ in the last $\lambda n_1 -n_1$ coordinates is a narrow-sense BCH code of length $n_1$. We prove that whenever $\mathcal{B}$ is Hermitian self-orthogonal, so is $\mathcal{C}$.

As a consequence, we present a procedure to obtain quantum stabilizer codes with lengths than cannot be reached by BCH codes. With this procedure we get new quantum codes according to Grassl's table \cite{codetables}.

To prove our results, we give necessary and sufficient conditions for  Hermitian self-orthogonality of BCH codes of a wide range of lengths.

\end{abstract}
	
\maketitle

\section{Introduction}

The idea of building efficient computers based in quantum mechanics arose towards the end of the 20th century. Facts like, although quantum states cannot be replicated, quantum error-correction can be performed \cite{23RBC}, or Shor quantum algorithms for prime factorization and computation of discrete logarithms \cite{Shor1}, fueled interest in quantum computing. In the last years, several sources have shown progress in the construction of quantum devices \cite{Aru, BallP, Zhong, LiuY, Natu} confirming the importance of quantum error-correction.

Quantum error-correcting codes (QECCs) were first introduced and studied in the binary case \cite{18kkk, Calderbank, Gottesman}. In a second step, QECCs were considered for the nonbinary case \cite{AK,Ketkar}.

An important family of QECCs are the so-called {\it quantum stabilizer codes}. Setting $q$ a prime power, these codes can be defined as the set of common eigenvectors (with respect to some eigenvalue) of elements in a commutative subgroup of the group generated by a nice error basis on $\mathbb{C}^{q^n}$, where $\mathbb{C}$ are the complex numbers and $n$ the length of the code. A QECC $\mathcal{Q}$ has parameters $[[n,k,d]]_q$ whenever $\mathcal{Q}$ is a linear subspace of $\mathbb{C}^{q^n}$ of dimension $q^k$ and all errors in $\mathcal{Q}$ with weight less than $d$ can be detected or have no effect on $\mathcal{Q}$, while some error with weight equal to $d$ cannot be detected. The class of quantum stabilizer codes have a major advantage, which is that these codes can be constructed from self-orthogonal additive classical codes included in $\mathbb{F}_q^{2n}$,  being $\mathbb{F}_q$ the finite field with $q$ elements and the orthogonality with respect to a trace-symplectic form. When considering classical linear  codes, it suffices to use codes $\Co \subseteq \F_{q^2}^n$ which are self-orthogonal with respect to the Hermitian inner product, as we state in the following result.

\begin{thm}\cite{Aly,Ketkar}\label{col:stabherm}
Let $\Co$ be an $[n,k]$ linear code over $\F_{q^2}$. Denote by $\boldsymbol{x}$ a vector $(x_1, \ldots, x_n) \in \F_{q^2}^n$ and consider the Hermitian inner product on $\F_{q^2}^n$: $\boldsymbol{x} \cdot_h \boldsymbol{y} := \sum_{i=1}^{n} x_i y_i^q$. Set $\Co^{\perp_h}$ the Hermitian dual code of $\Co$, i.e., $\Co^{\perp_h} := \{\boldsymbol{x} \in \F_{q^2}^n \; : \; \boldsymbol{x} \cdot_h \boldsymbol{y}=0 \mbox{ for all } \boldsymbol{y} \in \Co \}$.

Assume that $\Co\subseteq \Co^{\perp_h}$. Then, there exists an $[[n,n-2k,\geq d^{\perp_h}]]_q$ quantum stabilizer code, where $d^{\perp_h}$ stands for  the minimum distance of $\Co^{\perp_h}$.
\end{thm}

Theorem \ref{col:stabherm} and a close result that uses dual-containing classical codes with respect to the Euclidean inner product (CSS procedure) give rise to most families  of quantum stabilizer codes appearing in the vast literature  on these codes, some of these papers are \cite{BE, AK, Ketkar, Aly, XingC, lag2, Anderson}.

Bose-Chaudhuri-Hoquengheim codes (BCH codes) were discovered in 1959-1960 and are very useful in classical coding theory. They can be regarded as cyclic codes and are introduced in most information theory books. A seminal article by Aly et al. \cite{Aly} provided a bound on the designed distance of narrow-sense  BCH codes which implies Hermitian dual containing. Thus, following Theorem \ref{col:stabherm} and considering Hermitian dual codes, good quantum stabilizer codes can be obtained. Some further contributions \cite{LMFL2013,KZT2013,lag2,HZC2015,XLGM2016,LLLM2017,YZKL2017,QZ2017,LLLG2017,ZSL2018} have produced new families of QECCs coming from  classical BCH, negacyclic or constacyclic codes improving the results in \cite{Aly} for particular code lengths.

Next, we recall some recent references in the direction of considering Hermitian dual containing BCH codes for supplying quantum stabilizer codes. For a start we cite \cite{SYW2019}, where it is only considered the primitive case. Also \cite{LLGW2019}, where the authors characterize Hermitian dual-containing BCH codes of length $n = \frac{q^{2s}-1}{a}$, $s\geq 3$ odd, $a$ divides $q^s+1$ and $3 \leq a \leq  2(q^2- q + 1)$. In \cite{LS2021}, it is characterized the case $s=4\eta$, $\eta$ odd, $q \equiv 1 \mod 4$ and $n = \frac{q^{2s}-1}{4(q^2-1)}$. We also mention \cite{XL2021}, whose Hermitian construction deals with lengths which are divisible by $(q^\alpha -1)$, $\alpha$ being the order of $q^2$ modulo $n$, and \cite{ZZ2021}, where a sufficient condition for  Hermitian dual containing is given for lengths $n = r(q^2- 1)$, $r$ being a positive integer. Finally, in \cite{L2022} the author studies the maximum designed distances of LCD Hermitian-dual containing BCH codes.

In this paper, we introduce a new family of $q^2$-ary classical codes which is formed by the codes that we name {\it homothetic-BCH}. These codes can be defined as subfield-subcodes over $\F_{q^2}$ of certain $q^{2s}$-ary, $s \geq 2$, evaluation codes which we name homothetic evaluation codes, see Definition \ref{definit}.

BCH codes can be regarded as subfield-subcodes of a family of evaluation codes of univariate polynomials called \{1\}-affine variety codes \cite{GH2015,Cas}. This way of looking at the BCH codes makes them able to be extended to subfield-subcodes of certain evaluation codes of polynomials in several variables, called $J$-affine variety codes \cite{GH2015, GGHR2017}. In this article we take this point of view, under which it is convenient to consider Hermitian self-orthogonality instead of dual containing to obtain QECCs. Thus, for us, a BCH code is the subfield-subcode over $\mathbb{F}_{q^2}$ of an evaluation code over $\mathbb{F}_{q^{2s}}$ defined by a union of cyclotomic cosets with respect to $q^2$ (see Definition \ref{AA}). When considering a narrow-sense situation, our BCH codes are Euclidean duals of the corresponding narrow-sense BCH codes defined as cyclic codes, and Hermitian dual containment of these last codes determines QECCs with the same parameters as Hermitian self-orthogonal BCH codes regarded as subfield-subcodes of evaluation codes.

Homothetic-BCH codes have lengths that cannot be achieved by univariate $\{1\}$-affine variety codes and a suitable puncturing of a homothetic-BCH code $\mathcal{C}$ gives rise to a narrow-sense BCH code $\mathcal{B}$ in such a way that, whenever $\mathcal{B}$ is Hermitian self-orthogonal, so is $\mathcal{C}$, see Theorem \ref{self}. This fact together with that of the dimension and minimum distance of a homothetic-BCH code can be bounded in a similar way to the BCH case (see Theorem \ref{bounds}) allow us to find {\it a large family of new Hermitian self-orthogonal codes providing new quantum codes}.

Let us state our main result, for which we first introduce the following notation. Set $s$ a positive integer and  $\Lambda_{a_0}, \Lambda_{a_1}, \ldots, \Lambda_{a_v}$ the cyclotomic cosets with respect to $q^2$ in $\frac{\mathbb{Z}}{(q^{2s}-1)\mathbb{Z}}$, where  $a_i$, $0 \leq i \leq v$, are the minimum nonnegative representatives of $\Lambda_{a_i}$. Pick $1 \leq \tau \leq v$ and set $\Delta_1(\tau) := \cup_{\ell=1}^{\tau} \Lambda_{a_{\ell}}$. Consider a positive integer $n_1$ such that it divides $q^{2s}-1$ and set $\Delta'_1(\tau) := \cup_{\ell=1}^{\tau'} a'_\ell$ the union of cyclotomic cosets in $\frac{\mathbb{Z}}{n_1\mathbb{Z}}$ obtained by reducing modulo $n_1$ the elements in $\Delta_1(\tau)$.

\begin{A}[Theorem \ref{mainth} in the paper]
{\it Let $q$ be a power of a prime $p$ and consider $s\geq 2$ a positive integer. Let $n_1$ be a positive integer that divides $q^{2s}-1$ as we are going to introduce, and define
    $$L:=\begin{cases}
			qn_2- \min\left\{ \floor*{\frac{(q-1)n_2}{q^{s-1}+1}},  \floor*{\frac{(q-1)n_2-1}{q^{s-1}}} \right\}-1, & \text{if } n_1=(q^s+1)n_2 \text{, } n_2\mid q^s-1\text{, } s \text{ even,}\\
			n_2-1 & \text{if } n_1=(q^s+1)n_2 \text{, } n_2\mid q^s-1\text{, } s \text{ odd,}\\
            q^{\frac{s+a}{2}}+q^{\frac{s-a}{2}}-2 & \text{if } n_1=(q^s-1)(q^a+1) \text{, } s> a\neq 0 \text{, } \frac{s+a}{2} \text{ odd,}\\
            2q^{\frac{s}{2 }}-3 & \text{if } n_1=2(q^s-1) \text{, } s> 0 \text{, } \frac{s}{2} \text{ odd,}\\
            q(q^{\frac{s+a}{2}}-q^a-1)-1 & \text{if } n_1=(q^s-1)(q^a+1) \text{, } s> a \text{, } \frac{s+a}{2} \text{ even,}\\ & \text{(except when $q=2$ and $a=s-2$).}
	\end{cases}
    $$
Set $\lambda$ another positive integer such that $\lambda\leq \frac{q^{2s}-1}{n_1}$ and $\lambda n_1 \nmid q^{2s}-1$.

Consider a set $\Delta_1(\tau)$ such that $a'_{\tau'} \leq L$. Then, the homothetic-BCH code given by  $\Delta_1(\tau)$ is a $[\lambda n_1,\,\leq \sum_{\ell=1}^\tau \# \Lambda_{a_\ell}]_{q^2}$ Hermitian self-orthogonal code whose Hermitian dual has minimum distance at least $a_{\tau+1}$ and, therefore,  there exists a
    $$[[\lambda n_1, \geq \lambda n_1-2\sum_{\ell=1}^\tau \# \Lambda_{a_\ell},\,\geq a_{\tau+1}]]_q$$
quantum stabilizer code.

Additionally, assume that $a'_{\tau'} \leq L$ and $p\mid \lambda$. Then, the homothetic-BCH code given by  $\Lambda_{a_0} \cup \Delta_1(\tau)$ is a $[\lambda n_1,\,\leq \sum_{\ell=0}^\tau \# \Lambda_{a_\ell}]_{q^2}$ Hermitian self-orthogonal code whose Hermitian dual has minimum distance at least $a_{\tau+1}+1$ and, therefore, there exists a
$$[[\lambda n_1, \geq \lambda n_1-2\sum_{\ell=0}^\tau \# \Lambda_{a_\ell},\,\geq a_{\tau+1}+1]]_q$$
quantum stabilizer code.}
\end{A}

Applying Theorem A, we are able to get new quantum codes for which there is no known construction according to \cite{codetables}, see Subsections \ref{421} and \ref{422}, where we give new binary and  $5$-ary QECCs. These examples are only a sample of those that can be found with our procedure, no exhaustive searching has been conducted. We also give new $8$-ary QECCs by using our way of seeing BCH codes (see Subsection \ref{423}).

As before mentioned, we can prove Hermitian self-orthogonality of homothetic-BCH codes from Hermitian self-orthogonality of narrow-sense BCH codes. Then, in Section \ref{HBCH}, to make our paper more complete, we give a condition, Condition (\ref{CondZ}), characterizing Hermitian self-orthogonality of narrow-sense BCH codes. Since the condition has to do with the solutions of a certain bivariate equation, we specialize our condition to four different lengths in  Propositions \ref{Case1}, \ref{Case2}, \ref{Case3} and \ref{Case4}. These propositions support our main result. Finally, in Remark \ref{REM}, we show that, on the one hand, Propositions \ref{Case3} and \ref{Case4} allow us to get a major improvement on the bound given in \cite[Theorem 13]{Aly} to get QECCs from Hermitian duality. On the other hand, we prove that our characterization of Hermitian self-orthogonality of BCH codes contemplates new cases being complementary to previous results on the subject.

Section \ref{Sect2} of the paper introduces homothetic-BCH codes and proves that they are Hermitian self-orthogonal when the corresponding punctured BCH codes are Hermitian self-orthogonal. Section \ref{HBCH} characterizes Hermitian self-orthogonality of several families of BCH codes. As a consequence, in Subsection \ref{Sect41}, we state our main result on QECCs coming from homothetic-BCH codes. We conclude the paper by giving, in Subsection \ref{Sect42}, new QECCs whose parameters have no known construction.

% and \cite{ZKZ2023} where they determine the maximum designed distance of Hermitian dual-containnig constacyclic narrow-sense BCH codes of length $\frac{q^{2m}-1}{a(q+1)}$ over $F_{q^2}$, where $m\geq2$ is even and $a > 1$ is a divisor of $q-1$.

\section{Homothetic-BCH codes}
\label{Sect2}

Let $\mathcal{D} \subseteq \mathbb{F}^N$ be a linear error-correcting code of length $N$ over a finite field $\mathbb{F}$. In this paper, a linear error-correcting code $\mathcal{E} \subseteq \mathbb{F}^{N'}$, $N' > N$, is said to be an {\it enlargement} of $\mathcal{D}$ whenever $\mathcal{D}$ is the projected set of $\mathcal{E}$ to the $N$ first coordinates, i.e. $\mathrm{pr}_{\{1, \ldots, N\}} (\mathcal{E}) = \mathcal{D}$, where $\mathrm{pr}_{\{1, \ldots, N\}}: \mathbb{F}^{N'} \rightarrow \mathbb{F}^N$ is the map $\mathrm{pr}_{\{1, \ldots, N\}} (x_1, \ldots, x_{N'}) = (x_1, \ldots, x_N)$. Equivalently, one can say that $\mathcal{D}$ is the puncturing of $\mathcal{E}$ in the last $N' - N$ coordinates.

We devote this section  to introduce a new family of evaluation codes which will be used to obtain new quantum stabilizer codes. Our family is formed by codes that enlarge BCH codes.

For our purposes it is useful to regard BCH codes as evaluation codes. Our first subsection recalls that fact.

\subsection{BCH codes as evaluation codes}

Along this paper, $p$ denotes a prime number and $q$ a power of $p$. Our supporting finite field is $\mathbb{F}_{q^{2s}}$, where $s$ is a positive integer $s \geq 2$, and this is because we are interested in subfield-subcodes and in quantum codes coming from Hermitian self-orthogonal codes.

Let $N$ be a divisor of $q^{2s} - 1 $, consider the polynomial ring $\F_{q^{2s}}[X]$ and its ideal $J(N) = \gen*{X^{N} -1}_{\F_{q^{2s}}[X]}$. Sometimes, when considering an ideal and the context is clear, we omit the subindex. Setting $U(N) = \{\beta_1, \ldots, \beta_N\}$ the zero set of $J(N)$, we define the evaluation map
$$\ev_{U(N)}: =\faktor{\F_{q^{2s}}[X]}{J(N)} \to \F_{q^{2s}}^{N} \textrm{, } \quad \ev_{U(N)}(h)=\left(h(\beta_1),\dots,h(\beta_{N})\right),$$
where $h$ stands for a class in $\faktor{\F_{q^{2s}}[X]}{J(N)}$. For simplicity, we also denote by $h$ the representative of minimum degree of the class $h$.

Consider the set $E(N) := \{0, 1, \ldots, N-1\}$ regarded as a set of representatives of the quotient ring $\Z / N\Z$ and fix a subset $\Delta \subseteq E(N)$. Then, the $\conj*{1}$-affine variety code $\Co_{\Delta}^{U(N)}$ is defined as
    $$\Co_{\Delta}^{U(N)}:=\gen*{\ev_{U(N)}(X^e) \mid e \in \Delta}=\ev_{U(N)}(\gen*{X^e \mid e\in \Delta})\subseteq \F_{q^{2s}}^{N},$$
where $\gen*{A}$ stands for the $\F_{q^{2s}}$-vector space generated by $A$. $\Co_{\Delta}^{U(N)}$ is a particular case of $J$-affine variety code as introduced in \cite{GGHR2017}; $J$-affine variety codes are designed to extend the above construction to evaluation of polynomials in several variables.

When $\Delta = \{0, \ldots, k-1\}$, we get an RS code $\Co_{\Delta}^{U(N)}$ with parameters $[N,k,N-k+1]_{q^{2s}}$.

Following \cite{Bier, Cas}, one can regard BCH codes over $\F_{q^{2}}$ as subfield-subcodes of $\conj*{1}$-affine variety codes. Indeed, recalling that $E(N)$ is a set of representatives of $\Z / N\Z$, the {\it cyclotomic cosets in $E(N)$ with respect to $q^2$} are the sets of the form
$$\Lambda=\conj{q^{2i}e \mid i\geq 0}\subseteq E(N)$$
for some element $e\in E(N)$, where the products are modulo $N$. A set of this type is usually represented as $\Lambda_a$, where $a$ is the minimum element in $\Lambda$ with respect the usual order in $\N$. Denote by
    $$\mathcal{A}(N)=\conj*{a_0(N)=0<a_1(N)<\cdots<a_{\nu}(N)}\subseteq E(N)$$ the ordered set of representatives of all cyclotomic cosets. We will dispense the term $(N)$ when it is clear from the context.

Let us state the definition of BCH code we are going to use.
\begin{mydef}
\label{AA}
The $q^2$-ary BCH code of length $N$ given by the pair $(t,t')$, $0 \leq t \leq t' \leq \nu$, $\mathrm{BCH}^N(t,t')$, is the subfield-subcode $\mathcal{C}^{U(N)}_{\Delta^{U(N)} (t,t')} \cap \mathbb{F}_{q^2}^N$, where $\Delta^{U(N)} (t,t') = \cup_{\ell =t}^{t'} \Lambda_{a_\ell}$.
\end{mydef}

By \cite{GH2015,GGHR2017}, one gets the following well-known properties.

\begin{prop}
\label{BCH}
The following statements hold.
\begin{enumerate}
\item The code $\mathrm{BCH}^N(t,t')$ is a $q^2$-ary code with length $N$ and dimension $\sum_{\ell=t}^{t'} \# \Lambda_{a_\ell}$.
\item The minimum distance $d$ of the Hermitian dual $[\mathrm{BCH}^N(1,t)]^{\perp_h}$ (respectively,\\ $[\mathrm{BCH}^N(0,t)]^{\perp_h}$) of the code $\mathrm{BCH}^N(1,t)$ (respectively, $\mathrm{BCH}^N(0,t)$) satisfies $d \geq a_{t+1}$ (respectively, $d \geq a_{t+1} +1$).

\end{enumerate}
\end{prop}

\subsection{Homothetic evaluation and homothetic-BCH codes}
\label{SS22}
We are going to introduce the before announced new family of evaluation codes. Let $q$ and $s$ be integers as in the above subsection and fix  a positive integer $n_1$ which divides $q^{2s} - 1 $.  Consider also a positive integer $\lambda$ such that $\lambda\leq \frac{q^{2s}-1}{n_1}$ and $\lambda n_1 \nmid q^{2s}-1$. Let $\gamma\in\F_{q^{2s}}$ be a primitive element of the field $\F_{q^{2s}}$ and also take a primitive $n_1$-th root of unity, $\zeta_{n_1}\in\F_{q^{2s}}$.

Set
$$I=\gen*{(X^{n_1}-1)(X^{n_1}-\gamma^{n_1})\cdots (X^{n_1}-\gamma^{(\lambda-1)n_1})}_{\F_{q^{2s}}[X]}$$
the ideal of the ring of univariate polynomials in $\F_{q^{2s}}[X]$ with vanishing set (in $\F_{q^{2s}}$) $$P=\conj*{\alpha_1,\dots,\alpha_n}=
\conj*{1,\zeta_{n_1},\dots,\zeta_{n_1}^{n_1-1},\gamma,\gamma\zeta_{n_1},
\dots,\gamma\zeta_{n_1}^{n_1-1},\dots,\gamma^{\lambda-1},\gamma^{\lambda-1}\zeta_{n_1},\dots,\gamma^{\lambda-1}\zeta_{n_1}^{n_1-1}}.$$

It is clear that $P \subset U(q^{2s}-1) = \{\beta_1, \ldots, \beta_{q^{2s}-1}\} = \F_{q^{2s}}\backslash \conj*{0}$ and we take the convention that the elements in $U(q^{2s}-1)$ are ordered in such a way that the first ones are those in $P$; that is $\beta_1 = \alpha_1, \ldots, \beta_n = \alpha_n$. The same happens with $U(n_1) \subset P$.

Let
    $$\mathcal{R}=\faktor{\F_{q^{2s}}[X]}{I}$$
be the quotient ring by $I$ and consider the linear evaluation map
    $$
    \ev_P: \mathcal{R} \paren*{=\faktor{\F_{q^{2s}}[X]}{I}} \to \F_{q^{2s}}^{n} \textrm{, } \quad \ev_P(f)=\paren{f(\alpha_1),\dots,f(\alpha_n)},
    $$
where, as above, $f\in\mathcal{R}$ denotes an equivalence class in $\mathcal{R}$ and the unique polynomial in $\F_{q^{2s}}[X]$ with degree less than $\lambda n_1$ representing $f$. It is clear that $f$ can be written in a unique way as a linear combination of monomials with exponents in the set $E:=\conj*{0,1,\dots,\lambda n_1-1}$.

We are ready to introduce the families of codes we are interested in.

\begin{mydef}
\label{definit}
Let $\Delta \subseteq E(q^{2s}-1)$. The $q^{2s}$-ary {\it homothetic evaluation code} given by $\Delta$ is defined to be the following vector subspace of $\F_{q^{2s}}^n$:
    $$\mathcal{H}_\Delta^P:=\gen*{\ev_P(X^e) \mid e\in\Delta}=\ev_P(\gen*{X^e  \mid e\in \Delta})\subseteq \F_{q^{2s}}^n.$$
Note that $X^e$ also represents $X^e + I$. The subfield-subcode over $\F_{q^{2}}$
\[
\Su_\Delta^P := \mathcal{H}_\Delta^P \bigcap \F_{q^2}^n
\]
is named the {\it homothetic-BCH} $q^2$-ary code given by $\Delta$.
\end{mydef}

\begin{remark}
{\rm
The length of the code $\mathcal{H}_\Delta^P$ is $n=\lambda n_1$. We imposed the condition $\lambda n_1 \nmid q^{2s}-1$ so that the length $n$ cannot be obtained with a univariate $\{1\}$-affine variety code because these codes have lengths $n'$ satisfying $n' \mid q^{2s}-1$.

When $\Delta \subseteq E(n_1)$, the projected set $\mathrm{pr}_{\{1, \ldots, n_1\}} \mathcal{H}_\Delta^P$  of $\mathcal{H}_\Delta^P$ to the first $n_1$ coordinates coincides with the code $\Co_{\Delta}^{U(n_1)}$ and, then, $\mathcal{H}_\Delta^P$ can be seen as an enlargement of $\Co_{\Delta}^{U(n_1)}$. %In addition, when $\Delta \nsubseteq E(n_1)$, but $\Delta \subseteq E(\lambda n_1)$, the projection $\mathrm{pr}_{\{1, \ldots, n_1\}} \mathcal{H}_\Delta^P$ is a $\{1\}$-affine variety code.
}
\end{remark}

With the above notation, consider the set $$\mathcal{A}(q^{2s}-1)=\conj*{a_0(q^{2s}-1)=0<a_1(q^{2s}-1)<\cdots<a_{\nu}(q^{2s}-1)}$$ and fix a index $\tau \leq \nu$. Define
\begin{equation}
\label{eldelta}
\Delta_1(\tau) : = \Delta^{U(q^{2s}-1)}(1,\tau) \mbox{  and  } \Delta_0(\tau) : = \Delta^{U(q^{2s}-1)}(0,\tau).
\end{equation}
Then, reasoning as in \cite[Theorem 13]{Traza}, we get the next result.

\begin{thm}\label{bounds}
The following bounds hold on the dimension $\dim$ and the minimum distance $\dis$ of the Hermitian dual of the homothetic-BCH codes $\Su_{\Delta_1(\tau)}^P$ and $\Su_{\Delta_0(\tau)}^P$.
\begin{enumerate}
    \item $\dim \paren*{\Su_{\Delta_1(\tau)}^P} \leq \sum_{\ell=1}^\tau \# \Lambda_{a_\ell}$;
    $\; \;\dim \paren*{\Su_{\Delta_0(\tau)}^P} \leq \sum_{\ell=1}^\tau \# \Lambda_{a_\ell} + 1$.
    \item $\dis\paren*{\paren*{\Su_{\Delta_1(\tau)}^P}^{\perp_h}} \geq a_{\tau+1}$;
    $\; \; \dis\paren*{\paren*{\Su_{\Delta_0(\tau)}^P}^{\perp_h}} \geq a_{\tau+1} + 1$.
\end{enumerate}
\end{thm}

\subsection{Hermitian self-orthogonality of homothetic-BCH codes}
We devote this subsection to study when a code of the type $\Su_{\Delta_1(\tau)}^P$ is self-orthogonal with respect to the Hermitian inner product.

%We start by considering the code $\Co_{\Delta_0(1)}^{U(q^{2s}-1)}$  and the subfield-subcode $\Su_{\Delta_0(1)}^{U(q^{2s}-1)} := \Co_{\Delta_0(1)}^{U(q^{2s}-1)} \cap \F_{q^2}^{q^{2s}-1}$.

Consider the map
$$
\mathcal{T} \colon \faktor{\F_{q^{2s}}[X]}{I} \to \faktor{\F_{q^{2s}}[X]}{I}, \quad \mathcal{T}(f)=f+f^{q^2}+\cdots+f^{q^{2(s-1)}}.
$$
Reasoning as in \cite[Propositions 4 and 5]{GH2015}, a codeword $\bs{c}=\ev_{P}(f)\in \mathcal{H}_{\Delta_1(\tau)}^{P}$, with $f \in \faktor{\F_{q^{2s}}[X]}{I}$, satisfies $\bs{c} \in \Su_{\Delta_1(\tau)}^{P}$ if  $f=\mathcal{T}(g)$ for some $g\in \faktor{\F_{q^{2s}}[X]}{I}$.

To decide about self-orthogonality of $\Su_{\Delta_1(\tau)}^P$, it suffices to study the following inner product
$$\ev_P(X^e )\cdot_h \ev_P(X^{e'})$$
for $e,\,e' \in \conj*{0,1,\dots, q^{2s}-1}$.

Let us state our main result in this section.

\begin{thm}\label{self}
Keep the above notation and consider a set $\Delta_1(\tau)$ as defined above Theorem \ref{bounds}. Define $\Delta'_1(\tau')$ the set formed by the union of the cyclotomic cosets in $E(n_1)$, with respect to $q^2$, obtained by reducing the elements in  $\Delta_1(\tau)$  modulo $n_1$. That is
\begin{equation}
\label{eldeltap}
\Delta'_1(\tau')=\Lambda'_{a_1'} \cup \cdots \cup \Lambda'_{a'_{\tau'}}\subseteq \conj*{0,1,\dots,n_1-1}.
\end{equation}
Consider the equation
\begin{equation}\label{eqst}
     qx + q^{2k} y = \beta n_1, \; \; 0 \leq k \leq s-1, 0 < \beta \in \mathbb{Z},
\end{equation}
where $x$ and $y$ are integer variables such that $0 \leq x, y \leq n_1 -1,$ representing elements in $\frac{\mathbb{Z}}{n_1\mathbb{Z}} = E(n_1)$.

Then, if
\begin{equation}
\label{cota}
a'_{\tau'} \leq \min \left\{ \max \left\{ x,y \right\} \; | \; (x,y) \mbox{ is a solution of (\ref{eqst})} \right\} - 1,
\end{equation}
the $\{1\}$-affine variety code $\Co_{\Delta'_1(\tau')}^{U(n_1)}$ and the BCH code of length $n_1$ over $\mathbb{F}_{q^2}$,  $\mathrm{BCH}^{n_1}(1,\tau')=\Co_{\Delta'_1 (\tau')}^{U(n_1)} \cap \mathbb{F}_{q^2}^{n_1}$, are Hermitian self-orthogonal and, as a consequence,  $\mathcal{H}_{\Delta_1(\tau)}^P$ and $\Su_{\Delta_1(\tau)}^P$ are also Hermitian self-orthogonal codes.

\end{thm}

\begin{proof}
As we said before the statement, to get Hermitian self-orthogonality of the code $\Su_{\Delta_1(\tau)}^P$, it suffices to check that for any pair of exponents $e,\,e' \in \conj*{0,1,\dots, q^{2s}-1}$ of monomials to be evaluated to get $\mathcal{H}_{\Delta_1(\tau)}^{P}$, it holds that
$$\ev_P(X^e )\cdot_h \ev_P(X^{e'} ) = 0.$$
Since $P$ is the zero set of $I$, we are interested in the following chain of equalities:
\begin{align}\label{advsoc}
        \ev_P(X^e)\cdot_h \ev_P(X^{e'})&=\sum_{i=0}^{n_1-1} \zeta_{n_1}^{i(e+qe')}+\gamma^{e+qe'}\sum_{i=0}^{n_1-1} \zeta_{n_1}^{i(e+qe')}+\cdots+\gamma^{(\lambda-1)(e+qe')}\sum_{i=0}^{n_1-1} \zeta_{n_1}^{i(e+qe')} \nonumber \\
        &=\paren*{\sum_{i=0}^{n_1-1} \zeta_{n_1}^{i(e+qe')}} \paren*{1+\gamma^{e+qe'}+\cdots+\gamma^{(\lambda-1)(e+qe')}}.
    \end{align}

The chain of equalities (\ref{advsoc}) proves that if the factor $\sum_{i=0}^{n_1-1} \zeta_{n_1}^{i(e+qe')}$ vanishes, then we get the desired self-orthogonality. It is clear from our construction that if the mentioned code $\Co_{\Delta'_1 (\tau')}^{U(n_1)} \cap \mathbb{F}_{q^2}^{n_1}$ is Hermitian self-orthogonal, then the above factor is zero.

It remains to describe the role of equation (\ref{eqst}) and the fact that self-orthogonality holds whenever (\ref{cota}) happens. We have that
	\begin{equation*}
		\sum_{i=0}^{n_1-1} \zeta_{n_1}^{i(e+qe')}=
		\begin{cases}
			0, & \text{if } e+qe'\nequiv 0 \mod n_1,\\
			\neq 0 & \text{if } e+qe' \equiv 0 \mod n_1.
		\end{cases}
	\end{equation*}
	
{\it Notice that for $e=e'=0$, the above sum does not vanish and this is the reason we do not include $0$ in $\Delta_1(\tau)$.}
%However, when $p\mid \lambda$ and $e=e'=0$ the second factor of \eqref{advsoc} vanishes, and we could consider $0$ to belong to $\Delta$.

We desire to study when
\begin{equation}\label{geq}
e+qe'=\alpha n_1
\end{equation}
for some $\alpha\in \Z$, $\alpha>0$. In this stage, we are considering evaluation only in the $n_1$ first coordinates and, thus, we can suppose that $e$ and $e'$ belong to certain cyclotomic cosets  modulo $n_1$ with respect to $q^2$. Then, set
$$e=j_0q^{2t_1}, \quad e'=i_0q^{2t_2}, \quad t_1,\,t_2\in \N, \quad t_2\leq t_1\leq s-1,$$
where $i_0$ and $j_0$ denote the representatives in $\mathcal{A}(n_1)$ of the cyclotomic cosets to which $e$ and $e'$ belong, respectively.

Thus, we get the following equality
\begin{equation}
\label{TTT}
e + qe' = qe' + e = qi_0q^{2t_2}+j_0q^{2t_1}=q^{2t_2}\paren*{qi_0+q^{2(t_1-t_2)}j_0}=\alpha n_1.
\end{equation}
Since $n_1 \nmid q^{2t_2}$, the above equation is equivalent to
\begin{equation}
\label{la7}
    qi_0+q^{2k}j_0=\beta n_1
\end{equation}
with $k=t_1-t_2\in \N$ ($0\leq k\leq s-1$) for some $\beta\in \Z$, $\beta>0$.
Notice that we can assume $t_2\leq t_1$ because otherwise, when multiplying Equality (\ref{TTT}) by $q$, we would obtain $q^2i_0q^{2t_2}+qj_0q^{2t_1}=\alpha'n_1$. Then, it holds $i_0q^{2(t_2+1)}+qj_0q^{2t_1}=\alpha'n_1$, with arbitrary $t_2+1>t_1$, getting an analogous situation.

Therefore,
\begin{equation}
\label{MM}
L:= \min\conj*{\max\conj*{i_0,\,j_0} \mid (i_0,j_0)\in \mathcal{A}(n_1) \times \mathcal{A}(n_1)  \text{ is a solution of \eqref{eqst}}}-1
\end{equation}
is the bound we look for, because $L+1$ is  the first representative in $\mathcal{A}(n_1)$ we must avoid to get self-orthogonality. It means that {\it the bound is sharp, in the sense that if $a'_{\tau'}$ exceeds the bound, the corresponding BCH code will not be Hermitian self-orthogonal.}

Next we are going to prove that in order to obtain our bound, we do not need to restrict our computations to elements in $\mathcal{A}(n_1) \times \mathcal{A}(n_1)$ but it suffices to consider what we indicated in (\ref{cota}). Multiplying Equation \eqref{la7} by $q^2$, we get an equivalent equation. Recall that the solutions are elements in $\frac{\mathbb{Z}}{n_1\mathbb{Z}}$ and, therefore, the solutions shift within the same cyclotomic coset. This allows us to study Equation \eqref{eqst} by looking for solutions in $E(n_1)$, no matter if they are representatives of their cyclotomic cosets, because in our set of solutions, the solution made of representatives will also appear and it will determine the bound.  Indeed, if $(i,j)$ is a solution where $i \neq i_0$ and $j \neq j_0$, then $\max\{i,j\} > \max \{i_0,j_0\}$. In addition, in case that $(i_0,j)$ is a solution and $j$ is not the representative of its cyclotomic coset, then $(i_0,j)=(i_0,j_0q^{2t'})$ and $\max\{i_0,j_0q^{2t'}\} \geq \max \{i_0,j_0\}$. Analogously if $(i,j_0)$ is a solution and $i$ is not the representative of its cyclotomic coset, then $(i,j_0)=(i_0q^{2t'},j_0)$ and $\max\{i_0q^{2t'},j_0\} \geq  \max \{i_0,j_0\}$.

Therefore, the quantity
\begin{equation}
\label{CondZ}
\min\conj*{\max\conj*{x,\,y} \mid (x,y) \text{ is a solution of \eqref{eqst}}}-1,
\end{equation}
computed for general $(x,\,y) \in E(n_1) \times E(n_1)$, coincides with the sharp upper bound $L$ before given for the value $a'_{\tau'}$ we are interested in.
This concludes the proof.

%That is, if $\{1,\dots,L\}\subseteq \Delta'$ and
%$$L\leq \min\conj*{\max\conj*{x,\,y} \mid (x,y) \text{ is a solution of \eqref{eqst}}}-1$$
%then $\Co_{\Delta'}^{U(n_1)}$ is Hermitian self-orthogonal.

\end{proof}

%\begin{remark}
%\label{onBCH}
%\end{remark}

\begin{remark}
{\rm
In Theorem \ref{self} and its proof, we used the fact that the code  $\mathcal{H}_{\Delta_1(\tau)}^P$ introduced before Theorem \ref{bounds} can be regarded as a punctured code of the code $\Co_{\Delta_1 (\tau)}^{U(q^{2s}-1)}$. Also $\Co_{\Delta'_1 (\tau')}^{U(n_1)}$ with $\Delta'_1 (\tau')$ as in (2) of Theorem \ref{self} is a punctured code of these last two  codes. Moreover, when $\Co_{\Delta'_1 (\tau')}^{U(n_1)}$ is Hermitian self-orthogonal, so is $\Su_{\Delta_1(\tau)}^P$.

Theorem \ref{self} implies that giving conditions for Hermitian self-orthogonality of BCH codes (as introduced in Definition \ref{AA}) implies giving conditions for Hermitian self-orthogonality of homothetic-BCH codes. In Section \ref{Sect4} we will take advantage of this fact to enlarge BCH codes and to get bounds for its dimension and minimum distance by
Theorem \ref{bounds}. This way we will obtain quantum stabilizer codes better than those appearing in the literature. Finally, we note that, as we mentioned in the proof of Theorem \ref{self}, we consider sets $\Delta_1(\tau)$, but if we add the condition $p | \lambda$, sets $\Delta_0(\tau)$ can also be used, obtaining one unit more in our bounds for dimension and distance. The specific main result will be stated in the forthcoming Theorem \ref{mainth}.
}
\end{remark}

\section{Hermitian self-orthogonality of narrow-sense BCH codes}
\label{HBCH}
Theorem \ref{self} proves that one has a Hermitian self-orthogonal homothetic-BCH code $\Su_{\Delta_1(\tau)}^P$ when
%the $\{1\}$-affine variety code $\Co_{\Delta'_1(\tau')}^{U(n_1)}$  and the associated
the BCH code $\Co_{\Delta'_1 (\tau')}^{U(n_1)} \cap \mathbb{F}_{q^2}^{n_1}$ is Hermitian self-orthogonal.

Thus, we only need to provide conditions for Hermitian self-orthogonality of codes $\mathrm{BCH}^{n_1}(1,\tau') = \Co_{\Delta'_1 (\tau')}^{U(n_1)} \cap \mathbb{F}_{q^2}^{n_1}$ where
$$
\Delta'_1(\tau')=\Lambda'_{a_1'} \cup \cdots \cup \Lambda'_{a'_{\tau'}}.
$$
With independence of the case of homothetic-BCH codes, if one considers a Hermitian self-orthogonal code $\mathcal{C} = \mathrm{BCH}^{n_1}(1,\tau')$, by Theorem \ref{col:stabherm}, the length and dimension of $\mathcal{C}$ and the minimum distance of its Hermitian dual  give the parameters of the attached quantum stabilizer code. An analogous comparable situation is studied in a number of papers \cite{Aly, LMFL2013,KZT2013,lag2,HZC2015,XLGM2016,LLLM2017,YZKL2017,QZ2017,LLLG2017,ZSL2018,SYW2019,LLGW2019,LS2021,XL2021,ZZ2021,L2022}, where Hermitian dual containment is considered for BCH codes seen as cyclic codes.

In this section we take advantage of the fact that (\ref{cota}) determines a sharp bound to get  Hermitian self-orthogonal narrow-sense BCH codes. This allows us to find bounds on $a'_{\tau'}$ implying Hermitian self-orthogonality for several families of BCH codes which, in many cases, improve the bounds on the designed distance providing quantum stabilizer codes.

We remind that $q$ is a prime power, $s \geq 2 $ is an integer and suppose that $n_1$ is a positive integer which divides $q^{2s} - 1 $. We study four cases of BCH codes determined by the factorization of their lengths $n_1$.

\begin{itemize}
\item Case 1: $n_1=(q^s+1)n_2$, $n_2\mid q^s-1$, $s$ even.
\item Case 2: $n_1=(q^s+1)n_2$, $n_2\mid q^s-1$, $s$ odd.
\item Case 3: $n_1=(q^s-1)(q^a+1)$, $s> a$, $\frac{s+a}{2}$ odd.
\item Case 4: $n_1=(q^s-1)(q^a+1)$, $s> a$, $\frac{s+a}{2}$ even, with the exception of the case when $q=2$ and $a=s-2$.
\end{itemize}

The following four propositions consider each one of the previous cases and provide the announced bound on $a'_{\tau'}$. All the proofs are supported in the fact proved in Theorem \ref{self} that {\it the bound given in (\ref{cota}), see (\ref{CondZ}), is a sharp bound for Hermitian self-orthogonality of BCH codes.} Note that sharp means that if the bound is exceed, Hermitian self-orthogonality cannot hold.

\begin{prop}
\label{Case1}
Keep the above notation and consider a BCH code $\mathcal{C}=\mathrm{BCH}^{n_1}(1,\tau')$ whose length $n_1$ satisfies the conditions given in Case 1. Then, $\mathcal{C}$ is Hermitian self-orthogonal if and only if
$$a_{\tau'}' \leq qn_2- \min\left\{ \floor*{\frac{(q-1)n_2}{q^{s-1}+1}},  \floor*{\frac{(q-1)n_2-1}{q^{s-1}}} \right\}-1.$$
\end{prop}
\begin{proof}
The pair $(e=q^{s+1}n_2, e'=n_2)$ is a solution of \eqref{geq} for $\alpha=q$ and thus $(i=n_2,j=qn_2)$ is a solution of \eqref{eqst} for $\beta=q$ and $k=\frac{s}{2}$. Then, it suffices to study whether there exists another solution $(i',j')$ of \eqref{eqst} such that $\max\conj*{i',j'}<qn_2$. Let us assume it exists, then it can be written as
\begin{equation*}
\begin{cases}
i'=n_2+u\\
j'=qn_2-v
\end{cases}
\text{ with } u,\,v\in\Z \text{ such that }
\begin{cases}
-n_2<u<(q-1)n_2\\
0<v<qn_2.
\end{cases}
\end{equation*}
We use an strict inequality since we cannot pick the monomial $X^0$. Then, by \eqref{eqst}, one gets
$$q(n_2+u)+q^{2k}(qn_2-v)=\beta q^sn_2+ \beta n_2,$$
or equivalently
\begin{equation}
\label{TT}
q^{2k+1}n_2+qn_2-q^{2k}v+qu=\beta q^sn_2+\beta n_2.
\end{equation}

Now we distinguish {\it four} possibilities depending on the value $2k+1$ ($1\leq 2k+1\leq 2s-1$). Our result will follow from the fact that only one of them can hold.
\begin{itemize}
    \item Suppose that $2k+1=s-w$, $1\leq w < s$, $s\neq 2$.
\end{itemize}
    We are going to show, by contradiction, that this possibility cannot occur. One has
        $$q^{s-w}n_2+qn_2-q^{s-w-1}v+qu=\beta q^s n_2+\beta n_2, \;\mbox{ which implies}$$
        $$qn_2(\beta q^{s-1}-q^{s-w-1}-1)=qu-q^{s-w-1}v-\beta n_2 \;\mbox{ and then,}$$
        $$qn_2=\frac{qu-q^{s-w-1}v-\beta n_2}{\beta q^{s-1}-q^{s-w-1}-1}.$$
        Here $-q^{s-w-1}v-\beta n_2<0$, $\beta q^{s-1}-q^{s-w-1}-1>0$ and $\frac{q}{\beta q^{s-1}-q^{s-w-1}-1}<1$, since $s>2$. It implies that $qn_2<u$, obtaining a contradiction to the fact that $u<(q-1)n_2$.
\begin{itemize}
    \item Assume now that $2k+1=s-w$, $1 \leq w < s$, $s=2$.
\end{itemize}
    This option is not possible either. Indeed, in this case $k=0$ and Equality (\ref{TT}) is
        $$2qn_2+qu-v=\beta q^2 n_2 + \beta n_2,$$  which is equivalent to  $$v=-(q^2+1)\beta n_2+2qn_2+qu.$$
        Then, $v>0$ if and only if $u>\frac{n_2(q^2\beta+\beta-2q)}{q}$ and it gives a contradiction. In fact, if $\beta\neq 1$ it holds that $\frac{n_2(q^2\beta+\beta-2q)}{q}\geq (q-1)n_2$, exceeding the upper bound on $u$. Otherwise, $\beta=1$ implies $\frac{n_2(q^2\beta+\beta-2q)}{q}=\frac{n_2(q-1)^2}{q}$ and there is no $u\in \Z$ such that $n_2(q-1)\frac{q-1}{q}<u<n_2(q-1)$.

\begin{itemize}
    \item The third possibility is $2k+1=s+1$.
\end{itemize}
Then, Equality (\ref{TT}) can be expressed as follows:
    \begin{equation}
    \label{SS}
    q(u-q^{s-1}v)=(\beta-q)(q^s+1)n_2.
    \end{equation}
    The value $\beta$ must be a multiple of $q$ and $q(u-q^{s-1}v)<(q^s+1)n_2$ because the ranges of $u$ and $v$ show that $qu-q^sv<q(q-1)n_2-q^s<(q^s+1)n_2$. Thus, Equality (\ref{SS}) holds if and only if $\beta=q$ and $u=q^{s-1}v$. Therefore we get new solutions:
    $$(i'=n_2+q^{s-1}v,\,j'=qn_2-v).$$

    Finally, we consider the remaining case:
\begin{itemize}
    \item $2k+1=s+w$, $3\leq w\leq s-1$, $w$ odd.
\end{itemize}
We are going to show that it gives no new solution of  Equation \eqref{eqst}. Here, Equality (\ref{TT}) is
    $$q \gamma + \delta q^{s+w-1}=\beta n_1,$$
    for suitable positive integers $\gamma$ and $\delta$ depending on $u$ and $v$, respectively, and, by multiplying by $q^{s-w+2}$, we get $q^{s-w+3} \gamma + q \delta=\beta'n_1$ (recall that $q^{2s}=1$). To finish, since $s-w+3\leq s$, we fall back into the above cases, therefore we obtain no new solutions.
%\end{itemize}

Hence, the candidate for the bound on $a'_{\tau'}$ is given by the maximums of the pairs $(i'=n_2+q^{s-1}v,\,j'=qn_2-v)$. Let us study them.

Assume that $qn_2-v\geq n_2+q^{s-1}v$. This inequality happens if and only if $(q-1)n_2\geq (q^{s-1}+1)v$, which is equivalent to $v \leq \floor*{\frac{(q-1)n_2}{q^{s-1}+1}}$. In this case, the bound on $a_{\tau'}'$ is $qn_2-v-1$ for the maximum possible value $v$ since it gives the minimum value we look for. We also know that $v\leq \floor*{\frac{(q-1)n_2-1}{q^{s-1}}}$, it follows from the fact that $u=q^{s-1}v<(q-1)n_2$. As a consequence, in the previous situation,  one deduces that our bound for Hermitian self-orthogonality is
$$a_{\tau'}' \leq qn_2- \min\left\{ \floor*{\frac{(q-1)n_2}{q^{s-1}+1}},  \floor*{\frac{(q-1)n_2-1}{q^{s-1}}} \right\}-1.$$

%\Helena{Aquí pasa que en algunos casos el pequeño es el floor de la izquierda, y en estos casos se cumple que $n_2\geq \frac{q^{s-1}+2}{q-1}$. Con un programita comparando con Aly me da que nuestra cota es mejor una unidad en todos los casos excepto en algunos en los que evaluamos en todos los puntos, pero ni siquiera en todos. Algunos cálculos con magma me dan que la nuestra es la d}

We are going to prove that this is the actual bound because the complementary case cannot give a lower bound. Indeed, suppose that $qn_2-v< n_2+q^{s-1}v$, which is equivalent to $v > \frac{(q-1)n_2}{q^{s-1}+1}$, and let us show the result. Assume, by contradiction, that
$$n_2+q^{s-1}v < qn_2-\min\left\{ \floor*{\frac{(q-1)n_2}{q^{s-1}+1}},  \floor*{\frac{(q-1)n_2-1}{q^{s-1}}} \right\},$$
which leads to
\begin{equation}
\label{HHH}
\overbrace{\frac{(q-1)n_2}{q^{s-1}+1}}^{A:=} < v < \overbrace{\frac{(q-1)n_2}{q^{s-1}}}^{B:=}  - \overbrace{\frac{\min\left\{ \floor*{\frac{(q-1)n_2}{q^{s-1}+1}},  \floor*{\frac{(q-1)n_2-1}{q^{s-1}}} \right\}}{q^{s-1}}}^{C:=}.
\end{equation}
%Now $C>0$ and $B-C-A<1$ because $B-A=\frac{(q-1)n_2}{q^{s-1}(q^{s-1}+1)}<1$ (since $n_2\mid q^s-1$). As a consequence, there is no $v\in \mathbb{Z}$ as above, obtaining the desired contradiction.
Let us show that there is no positive integer $v$ satisfying (\ref{HHH}) which gives the desired contradiction.

First, assume  $C= \floor*{\frac{(q-1)n_2}{q^{s-1}+1}} / q^{s-1}$, that is, $C$ is given by the first quotient in the minimum that defines it. Set $C'= \frac{(q-1)n_2}{(q^{s-1}+1)q^{s-1}}$. Then, $B-A-C'=0$, $C= \frac{\floor*{A}}{q^{s-1}}$ and $C'= \frac{A}{q^{s-1}}$. Now $$B-A-C = B-A-C' +(C'-C) = C'- C = \frac{A- \floor*{A}}{q^{s-1}} < 1.$$
Thus, $v$ will be an integer if $v = \floor*{A} +1$ and $\floor*{A} +1 -A < \frac{A- \floor*{A}}{q^{s-1}}$, which holds if and only if
$A - \floor*{A} > \frac{q^{s-1}}{q^{s-1} +1}$ and it is a contradiction because $q^{s-1} +1$ is the denominator of $A$.

Now, write  $D= \frac{(q-1)n_2-1}{q^{s-1}}$ and
\[
D-A = \frac{(q-1)n_2 - (q^{s-1}+1)}{q^{s-1} (q^{s-1} +1)}.
\]
It remains to consider the case when $C$ is given by the second quotient in the minimum that defines it, that is $C = \floor*{\frac{(q-1)n_2-1}{q^{s-1}}}/q^{s-1}$, which is true if and only if $D-A <0$. Thus, $(q-1)n_2 - (q^{s-1}+1) <0$ and then $(q-1)n_2 < q^{s-1}+1$. This implies that $C=0$ and from (\ref{HHH}), we deduce
\[
(q-1)n_2 < v (q^{s-1}+1) < (q-1)n_2 + \frac{(q-1)n_2 }{q^{s-1}}.
\]
The above equality cannot hold because $\frac{(q-1)n_2 }{q^{s-1}} \leq 1$. Indeed, if $\frac{(q-1)n_2 }{q^{s-1}} > 1$, then $q^{s-1} < (q-1)n_2  < q^{s-1}+1$, which is not possible since all the values are positive integers. This concludes the proof.

\end{proof}

Our next result considers the former Case 2, where the length of the BCH code is $n_1= (q^s +1)n_2$, $n_2 \mid  q^s-1$ and $s$ is odd.
\begin{prop}
\label{Case2}
Keep the above notation and consider the BCH code $\mathcal{C}=\mathrm{BCH}^{n_1}(1,\tau')$ whose length $n_1$ satisfies the conditions given in Case 2. Then, $\mathcal{C}$ is Hermitian self-orthogonal if and only if
$$a_{\tau'}' \leq n_2-1.$$
\end{prop}

\begin{proof}
The pair $(e=n_2,e'=q^{s-1}n_2)$ is a solution of \eqref{geq} for $\alpha=1$ and therefore $(i=n_2,j=n_2)$ is a solution of \eqref{eqst} for $\beta=q$ and $k=\frac{s+1}{2}$ (after multiplying the first equation by $q$). Then, it suffices to study whether there exists another solution  $(i',j')$  of \eqref{eqst} such that $\max\conj*{i',j'}<n_2$. If this is the case, then
the solution has the shape $(i',j')=(n_2-u,n_2-v)$ where
$0< u,v< n_2$. We are going to prove, by contradiction, that there is no such solution.

By substituting in \eqref{eqst}, we get
\begin{equation}
\label{NNN}
q(n_2-u)+q^{2k} (n_2-v)=\beta_1 n_1,
\end{equation}
where $\beta_1\geq q$. This follows from the fact that $k >0 $ because, otherwise, we would have
$q(n_2-u) + n_2 -v =(q+1)n_2 - u - v = \beta( q^s +1) n_2$ with $\beta \geq 1$, a contradiction.

Equality (\ref{NNN}) occurs if and only if
\begin{equation}\label{eq2.2}
(n_2-u)+q^{2k-1} (n_2-v)=\beta_2 n_1
\end{equation}
with $\beta_2=\frac{\beta_1}{q}\geq 1$.

We may assume that $k\leq \frac{s+1}{2}$ because, otherwise, multiplying  by $q^{2s-2k+1}$, we obtain an equivalent equation to \eqref{eq2.2}. In fact, after this multiplication, one gets
\begin{equation}\label{RR}
q^{2s-2k+1}(n_2-u)+n_2-v=\beta_3(q^s+1)n_2,
\end{equation}
with $\beta_3=q^{2s-2k+1}\beta_2\geq q$. Note that $k\leq s-1$ implies $2k\leq 2s-2<2s+1$ and $2s-2k+1>0$. Therefore, Equality (\ref{RR}) is a particular case of \eqref{eq2.2} because $u$ and $v$ play the same role and $k':=s-k< \frac{s+1}{2}$, since $k> \frac{s+1}{2}$.

Returning to Equality \eqref{eq2.2}, it is equivalent to
\begin{equation}
\label{La15}
-(u+q^{2k-1} v)=n_2(\beta_2 (q^s+1)-q^{2k-1}-1).
\end{equation}
Then, the facts that $u,v>0$, $\beta_2\geq 1$ and $k\leq \frac{s+1}{2}$ prove that the left and the right hand sides of Equality (\ref{La15}) have different signs with the exception of the case where $k=\frac{s+1}{2}$ and $\beta_2=1$. In this last case, the RHS vanishes while the LHS is negative. Thus, we get the desired contradiction.

Therefore, a necessary and sufficient condition for Hermitian self-orthogonality of $\mathcal{C}$ is
$$a_{\tau'}' \leq n_2-1.$$
\end{proof}

Let us study the Case 3. Recall that, in this case, $n_1 \mid q^{2s}-1$, $n_1=(q^s-1)(q^a+1)$, $s> a$ and $\frac{s+a}{2}$ odd.
\begin{prop}
\label{Case3}
Keep the above notation and consider the BCH code $\mathcal{C}=\mathrm{BCH}^{n_1}(1,\tau')$ whose length $n_1$ satisfies the conditions given in Case 3. Then, $\mathcal{C}$ is Hermitian self-orthogonal if and only if
$$a_{\tau'}'\leq q^{\frac{s+a}{2}}+q^{\frac{s-a}{2}}-2 \mbox{ when $a \neq 0$ (respectively, } a_{\tau'}' \leq 2q^{\frac{s}{2 }}-3 \mbox{ in case $ a=0$).}$$
\end{prop}

\begin{proof}
We divide the proof in two parts corresponding to the settings where $a\neq 0$ and $a=0$.\\

{\it Part A}: $a \neq 0$.

The pair $(i=q^{\frac{s+a}{2 }}-q^a-1,j=q^{\frac{s+a}{2}}+q^{\frac{s-a}{2}}-1)$ is a solution of \eqref{eqst} for $\beta=q$ and $2k-1=\frac{a+s}{2}$ (recall that $k\leq s-1$). Also $s>a$ implies $i>0$ and we notice that $s-a$ is even.

Assume that there exists another solution $(i',j')$ of \eqref{eqst} such that $\max\{i',j'\}<\max\{i,j\}=j$. Then, we can express our solution as $(i'=q^{\frac{s+a}{2 }}-q^a-1+u,j'=q^{\frac{s+a}{2}}+q^{\frac{s-a}{2}}-1-v)$,
where $-i+1\le u< q^{\frac{s-a}{2}}+q^a$ and $1\le v<j$. Our new solution must satisfy \eqref{eqst}. Therefore,
\begin{equation}
\label{MMM}
q(q^{\frac{s+a}{2 }}-q^a-1+u)+q^{2k}(q^{\frac{s+a}{2}}+q^{\frac{s-a}{2}}-1-v)=\beta (q^s-1)(q^a+1),
\end{equation}
where $\beta\geq q$, which holds because $k >0$. Note that, otherwise when $k=0$, we get
\begin{multline*}
	q(q^{\frac{s+a}{2 }}-q^a-1+u)+q^{2k}(q^{\frac{s+a}{2}}+q^{\frac{s-a}{2}}-1-v) < (q+1)j \\< (q+1) \left( q^{\frac{s+a}{2}}+q^{\frac{s-a}{2}} \right) < \beta \left( q^{\frac{s+a}{2}} + q^s - q^a -1 \right),
\end{multline*}
with $\beta \geq 1$, which taking into account the $q$-adic expansions and the range of values of $s$ and $a$ is false.

Equality (\ref{MMM}) is equivalent to
$$
q^{\frac{s+a}{2 }}-q^a-1+u+q^{2k-1}(q^{\frac{s+a}{2}}+q^{\frac{s-a}{2}}-1-v)=\beta_1 (q^s-1)(q^a+1)
$$
where $\beta_1=\frac{\beta}{q}\geq 1$. This is true if and only if
\begin{align}
    u&=\beta_1(q^s-1)(q^a+1)-q^{\frac{s+a}{2}}+(q^a+1)-q^{2k-1}(q^{\frac{s+a}{2}}+q^{\frac{s-a}{2}}-1-v)\\ \label{FER}
    &=\beta_1(q^{s+a}+q^s)-(\beta_1-1)(q^a+1)-q^{2k-1}(q^{\frac{s+a}{2}}+q^{\frac{s-a}{2}}-1-v)-q^{\frac{s+a}{2}}.
\end{align}
Let us see that the above equality contradicts the range of values of $u$ and, therefore, there is no solution $(i',j')$ as proposed.
To show it, we consider the following three possibilities:
\begin{itemize}
    \item First assume that $2k-1<\frac{s+a}{2}$.
\end{itemize}
Then, considering the $q$-adic expansion of $q^{2k-1}(q^{\frac{s+a}{2}}+q^{\frac{s-a}{2}}-1-v)$ for any $v$, it holds that $u > q^{s+a} > q^{\frac{s-a}{2}} +q^a$, which is a contradiction.
    %$C<\beta_1(q^{s+a}+q^s)$ and $u=O(\beta_1(q^{s+a}+q^s))>q^{\frac{s-a}{2}}+q^a$, giving a contradiction.
\begin{itemize}
    \item In a second step, suppose that $2k-1=\frac{s+a}{2}$.
\end{itemize}
Then, if $\beta_1=1$ we would have $u=q^{\frac{s+a}{2}}v$, and the range of values for $u$ and $v$ leads to a contradiction. If $\beta_1>1$, then $$u=q^{\frac{s+a}{2}}v+(\beta_1-1)(q^{s+a}+q^s)-(\beta_1-1)(q^a+1) $$ and, by considering the $q$-adic expansions of $u$ and  $q^{\frac{s-a}{2}} +q^a$, it is clear that $u \geq q^{\frac{s-a}{2}} +q^a$, getting again to a contradiction.

It remains to study the last possibility:
\begin{itemize}
\item $2k-1>\frac{s+a}{2}$.
\end{itemize}
We can write $2k-1=\frac{s+a}{2} + \alpha$, $\mathbb{N} \ni\alpha \geq 1$. Then, taking into account (\ref{FER}), one can express $u$ as follows:
    \begin{align*}
    u&=\beta_1\left(q^{s+a}+q^s\right)-(\beta_1-1)\left(q^a+1\right) -q^\alpha\left(q^{s+a}+q^s\right) + q^{\frac{s+a}{2}+\alpha}(1+v)- q^{\frac{s+a}{2}}\\
    &=\left(\beta_1-q^\alpha\right)\left(q^{s+a}+q^s\right)+q^{\frac{s+a}{2}+\alpha}(1+v)- \left[(\beta_1-1)(q^a+1)+q^{\frac{s+a}{2}}\right].
\end{align*}
Now, we divide again our study into two cases:

i) $\beta_1 > q^\alpha$. Then, we write $\beta_1 =q^\alpha + \beta_2$, $\beta_2 \geq 1$. Next we will see that, in this case, \begin{equation}
\label{CAR}
u \geq q^{\frac{s-a}{2}} +q^a,
\end{equation}
which contradicts the restriction for $u$ and therefore there is no solution $(i',j')$ as indicated. To prove it, it suffices to write $u$ in the form
\[
u= \beta_2 \left(q^{s+a}+q^s \right) + q^\alpha \left[(1+v) \, q^{\frac{s+a}{2}}-q^a-1 \right] - (\beta_2-1)\left(q^a+1 \right)- q^{\frac{s+a}{2}},
\]
and considering its $q$-adic expansion and that of $q^{\frac{s-a}{2}} + q^a$, it is clear that Inequality (\ref{CAR}) is true.

Let us see the complementary case.

ii) $\beta_1 \leq q^\alpha$. Set $\beta_1 = q^\alpha - \beta_2$, $\beta_2 \geq 0$. Then, as before,
\[
u= q^\alpha \left[ (1+v) q^{\frac{s+a}{2}}- \left( q^a +1 \right) \right] - \beta_2 \left[q^{s+a}+q^s- \left(q^a+1\right) \right] + q^a+1 - q^{\frac{s+a}{2}}.
\]
We are going to prove that $u$ cannot belong to the range $-i+1\le u< q^{\frac{s-a}{2}}+q^a$. To do it, we just have to see that
\[
M:= q^\alpha \left[ (1+v) q^{\frac{s+a}{2}}- \left( q^a +1 \right) \right] - \beta_2 \left[q^{s+a}+q^s- \left(q^a+1\right) \right]
\]
satisfies either

(I) $M \geq q^{\frac{s-a}{2}}+q^a + q^{\frac{s+a}{2}} - \left( q^a+1 \right) = q^{\frac{s+a}{2}} + q^{\frac{s-a}{2}} -1$,

or

(II) $M < -q^{\frac{s+a}{2}} + q^a +1 +1 + q^{\frac{s+a}{2}} -\left( q^a+1\right) = 1$.

To conclude the proof of this complementary case ii), we again consider two possible settings:

ii-a) First suppose that $q^\alpha \left[ (1+v) q^{\frac{s+a}{2}}- \left( q^a +1 \right) \right] \leq \beta_2 \left[q^{s+a}+q^s- \left(q^a+1\right) \right]$. Then $M \leq 0$ and the result is proved because (II) holds.

ii-b) Finally, assume that
\begin{equation}
\label{HCF}
q^\alpha \left[ (1+v) q^{\frac{s+a}{2}}- \left( q^a +1 \right) \right] > \beta_2 \left[q^{s+a}+q^s- \left(q^a+1\right) \right].
\end{equation}
Here $q^\alpha > \beta_2$ and, to prove (I),  it suffices to consider $v=1$ since $M$ increases as $v$ grows (notice that $\beta_2$ decreases when $v$ gets larger). On the one hand, the coefficients of the $q$-adic expansion of $q^{s+a}+q^s- \left(q^a+1\right)$ are of the form $$1,0, \ldots, 0, q-1, \ldots, q-1, q-2, q-1, \ldots, q-1,$$ where $1$ is the coefficient of $q^{s+a}$, the first coefficient $q-1$ corresponds to $q^{s-1}$ and the remaining coefficients are $q-1$ with the exception of that of $q^a$ which is $q-2$. On the other hand, the coefficients of the $q$-adic expansion of $2\; q^{\frac{s+a}{2}}- \left( q^a +1 \right)$, recall that $v=1$, are $$1,q-1, \ldots, q-1, q-2, q-1, \ldots, q-1,$$ where $1$ is the coefficient of $q^{\frac{s+a}{2}}$, $q-2$ that of $q^a$ and the remaining ones are $q-1$. Thus, whatever $\alpha$ is, The $q$-adic expansion of $q^\alpha \left[ 2 q^{\frac{s+a}{2}}- \left( q^a +1 \right) \right]$ is obtained by shifting that of  $2 q^{\frac{s+a}{2}}- \left( q^a +1 \right)$ in such a way that (\ref{HCF}) holds and, by paying attention  to the coefficient of $q^{\frac{s+a}{2}}$ in the $q$-adic expansion of $M$, we deduce that  (I) must also occur.
%\end{itemize}

Therefore, we have proved that
$$a_{\tau'}'\leq q^{\frac{s+a}{2}}+q^{\frac{s-a}{2}}-2$$
is an equivalent condition to Hermitian self-orthogonality of $\mathcal{C}$ when $a \neq 0$.
\vspace{2mm}

{\it Part B}: $a=0$.

The pair $(i=2q^{\frac{s}{2 }}-2,j=2q^{\frac{s}{2}}-2)$ is a solution of \eqref{eqst} for $\beta=q$ and $2k-1=\frac{s}{2}$. Again we study whether there exists another solution $(i',j')$ of \eqref{eqst}, such that $\max\conj*{i',j'}<2q^{\frac{s}{2 }}-2$. If it existed, then it would be of the form $(i',j')=(2q^{\frac{s}{2 }}-2-u,2q^{\frac{s}{2 }}-2-v)$, where
$0< u,v< 2q^{\frac{s}{2 }}-2$. As in the above proofs, reasoning by contradiction, we will prove that there is no such solution $(i',j')$.

Substituting in \eqref{eqst}, we get
\[
q(2q^{\frac{s}{2 }}-2-u)+q^{2k} (2q^{\frac{s}{2 }}-2-v)=\beta_1 n_1,
\]
where $\beta_1\geq q$ because $k>0$. Note that, as in the case $a \neq 0$, if $k=0$ we get a contradiction. The above displayed equality is true if and only if
\begin{equation}\label{eq2.3}
(2q^{\frac{s}{2 }}-2-u)+q^{2k-1} (2q^{\frac{s}{2 }}-2-v)=\beta_2 n_1
\end{equation}
with $\beta_2=\frac{\beta_1}{q}\geq 1$. %We may assume that $k\leq \frac{s}{2}$ \Helena{NO, faltan casos} because, otherwise, we could obtain an equivalent equation to \eqref{eq2.3} by multiplying it by $q^{2s-2k+1}$ (note that $k\leq s-1$ implies $2k\leq 2s-2<2s+1$ so $2s-2k+1>0$):
%$$q^{2s-2k+1}(2q^{\frac{s}{2 }}-2-u)+2q^{\frac{s}{2 }}-2-v=\beta_32(q^s-1)$$
%with $\beta_3=q^{2s-2k+1}\beta_2>0$, and this equation is a particular case of \eqref{eq2.3}. This is because $u$ and $v$ play the same role and $k':=s-k< \frac{s}{2}$, since $k> \frac{s}{2}$.
Notice that \eqref{eq2.3} is equivalent to
\begin{equation}
\label{SSS}
-(u+q^{2k-1} v)=\beta_22(q^s-1)-(2q^{\frac{s}{2}}-2)(q^{2k-1}+1).
\end{equation}
Assume that $k \leq \frac{\frac{s}{2}+1}{2}$, then the facts that $u,v>0$ and $\beta_2\geq 1$ show that the two sides of Equality (\ref{SSS}) have different sign, except in the case where $k=\frac{\frac{s}{2}+1}{2}$ and $\beta_2=1$. In this last case the RHS of (\ref{SSS}) vanishes while the LHS is negative. Therefore, one gets a contradiction.

If, otherwise, $\frac{\frac{s}{2}+1}{2} <  k = \frac{\frac{s}{2}+1+s_1}{2}$, with $1\leq s_1\leq \frac{3s-3}{2}$ (so that $k\leq s-1$), then we get
\begin{equation}
\label{KK}
-(u+q^{\frac{s}{2}+s_1} v)=\beta_22(q^s-1)-(2q^{\frac{s}{2}}-2)(q^{\frac{s}{2}+s_1}+1).
\end{equation}
Assume first that $\beta_2\geq q^{s_1}$, then the RHS of (\ref{KK}) can be written as the difference of two positive integers and its $q$-adic expansion either starts with a power of $q$ larger than or equal to $s+s_1$ or with $2q^{\frac{s}{2}+s_1}$. Therefore, the RHS  of (\ref{KK}) is positive, while the LHS is negative, providing a contradiction.

It remains to consider the situation where $\beta_2< q^{s_1}$. Here (\ref{KK}) can be expressed as
\[
2\left(q^{\frac{s}{2}}-1\right) \left( q^{\frac{s}{2}+s_1} +1 \right) = \beta_2 2 \left( q^s -1\right) + u + q^{\frac{s}{2}+s_1} v
\]
and then, the facts that $u,\,v< 2q^{\frac{s}{2}}-2$  and the $q$-adic expansion of both terms in the equality, supply again a contradiction.

As a consequence, we obtain the bound
$$a_{\tau'}' \leq 2q^{\frac{s}{2 }}-3$$
given in the statement.

%\end{itemize}

\end{proof}

To finish this section, we give a bound on $a_{\tau'}'$ for the Case 4. Recall that this case corresponds to $n_1 \mid q^{2s}-1$, $n_1 = \left(q^s-1\right) \left(q^a +1\right)$, $s >a$ and $\frac{s+a}{2}$ even, where we exclude the situation when $q=2$ and $a= s-2$.
\begin{prop}
\label{Case4}
Keep the above notation and consider the BCH code $\mathcal{C}=\mathrm{BCH}^{n_1}(1,\tau')$ whose length $n_1$ satisfies the conditions given in Case 4. Then, $\mathcal{C}$ is Hermitian self-orthogonal if and only if
$$a_{\tau'}' \leq q(q^{\frac{s+a}{2}}-q^a-1)-1.$$
\end{prop}
\begin{proof}
The pair  $(i=q(q^{\frac{s+a}{2}}-q^a-1),j=q^{\frac{s+a}{2}}+q^{\frac{s-a}{2}}-1)$ is a solution of \eqref{eqst} for $\beta=q^2$ and $2k-2=\frac{s+a}{2}$ (remind that $k\leq s-1$).

Also $s>a$ implies $i>0$. Following the same line of the proofs of the results in this section, we assume the existence of another solution $(i',j')$ of \eqref{eqst} such that $\max\{i',j'\}< \max\{i,j\}$ and, by contradiction, we will prove that it is not possible.

We avoid the case $q =2$ and $s = a+2$. For any other $q,\; s$ and $a$ in this Case 4, $i = \max\{i,j\}$.  In fact, $i >j$ is equivalent to $$q^{\frac{s+a}{2} +1} > q^{\frac{s+a}{2}} + q^{a+1} + q^{\frac{s-a}{2} } + q -1$$ and, considering the involved $q$-adic expansions, it always holds with the exception of the case $ \frac{s+a}{2}= a+1$, which means that both equalities, $s=a+2$ and $q=2$, hold.

The possible solution $(i',j')$ must satisfy $$(i'=q(q^{\frac{s+a}{2 }}-q^a-1)-u,j'=q^{\frac{s+a}{2}}+q^{\frac{s-a}{2}}-1+v),$$
where $1\le u<i$ and $-j+1\le v\leq q^{\frac{s+a}{2}+1}-q^{\frac{s+a}{2}}-q^{a+1}-q^{\frac{s-a}{2}}-q$.
Then it must fulfill \eqref{eqst}:
$$
q(q(q^{\frac{s+a}{2 }}-q^a-1)-u)+q^{2k}(q^{\frac{s+a}{2}}+q^{\frac{s-a}{2}}-1+v)=\beta (q^s-1)(q^a+1),
$$
where $\beta\geq q$ because $k >0$. Note, that reasoning in a similar way to the proof of Proposition \ref{Case3}, the case $k=0$ leads to a contradiction. The above equality is equivalent to
$$
q(q^{\frac{s+a}{2 }}-q^a-1)-u+q^{2k-1}(q^{\frac{s+a}{2}}+q^{\frac{s-a}{2}}-1+v)=\beta_1 (q^s-1)(q^a+1)
$$
where $\beta_1=\frac{\beta}{q}\geq 1$. This happens if and only if
\begin{align*}
    u&=-\beta_1(q^s-1)(q^a+1)+q^{2k-1}(q^{\frac{s+a}{2}}+q^{\frac{s-a}{2}}-1+v)+q^{\frac{s+a}{2}+1}-q^{a+1}-q\\
    &=-\beta_1(q^{s+a}+q^s)+\beta_1(q^a+1)+q^{2k-1}(q^{\frac{s+a}{2}}+q^{\frac{s-a}{2}}-1+v)+q^{\frac{s+a}{2}+1}-q^{a+1}-q.
\end{align*}
Let us see that this contradicts the range of values of $u$, which concludes the proof. Consider two cases:
\begin{itemize}
    \item First assume that $2k-1\leq \frac{s+a}{2}-1$.
\end{itemize}
Then, $u \leq A + B$, where $A= (1- \beta_1)\left( q^{s+a} +q^s \right) + \beta_1 \left( q^a +1\right)$ and $$B= q^{\frac{s+a}{2}+1} - \left( q^{\frac{s+a}{2}} + q^{\frac{s+a}{2}-1} + q^{\frac{s+3a}{2}} + q^{s} + q^{a+1} +q \right).$$
    Now, we consider two possibilities.

    The first one is that $\beta_1 >1$, then clearly $u <0$. Otherwise $\beta_1 =1$, which implies $u \leq (q^a +1) + B$ and this last value is negative because $a < \frac{s+a}{2}$ and $\frac{s+a}{2} +1 \leq s$. As a consequence, we get again that $u <0$, which contradicts the fact $u \geq 1$.

\begin{itemize}
    \item The remaining case corresponds to the inequality $2k-1\geq \frac{s+a}{2}$.
\end{itemize}
Set $C:=q^{2k-1}(q^{\frac{s+a}{2}}+q^{\frac{s-a}{2}}-1+v)$ and recall that $i= q^{\frac{s+a}{2}+1}-q^{a+1}-q$. We again divide this case into two subcases.
        \begin{itemize}
            \item  $C\geq \beta_1(q^{s+a}+q^s)$. Then, one can write $C$ as $C=\beta_1(q^{s+a}+q^s)+C'$ with $C'\geq 0$. This implies $u=C'+\beta_1(q^a+1)+i>i$, supplying a contradiction.
            \item  Otherwise, $C<\beta_1(q^{s+a}+q^s)$, where we notice that $\beta_1$ could be greater than $q$. Then, we express $C$ as  $C=Q(q^{s+a}+q^s)+R$ with $Q\geq 0$, $0\leq R< q^{s+a}+q^s$. It implies $Q<\beta_1$ and $u=-(\beta_1-Q)(q^{s+a}+q^s)+R+\beta_1(q^a+1)+i$. Again, we consider two possibilities.

                The first one corresponds to the inequality $(\beta_1 -Q) \left(q^{s+a} + q^s\right) > \beta_1 (q^a +1)$. It can be equivalently written $\left[ (\beta_1 -Q) q^s - \beta_1 \right] \left(q^a +1\right) > 0$ and also
                \begin{equation}
                \label{FIN}
                \left[ \beta_1 (q^s -1) - Q q^s \right] \left(q^a +1\right) > 0.
                \end{equation}
                Clearly, $\beta_1 >Q$. Since the coefficients of the $q$-adic expansions of $q^s -1$ and $q^s$ are $0, q-1, \ldots, q-1$  and $1, 0, \ldots, 0$, respectively (the first coefficient corresponds to $q^{s}$), we deduce that the coefficient of $q^{s+a}$ (or, may be, of a larger exponent of $q$) in the $q$-adic expansion of the LHS of (\ref{FIN}) is not zero. Then, considering the $q$-adic expansion of $R+i$, we deduce $u <0$ giving the wanted contradiction.

                The second possibility holds whenever $(\beta_1 -Q) \left(q^{s+a} + q^s\right) \leq \beta_1 (q^a +1)$, and in this case $u \geq i$, which is also a contradiction.
        \end{itemize}

Therefore, we have proved that the bound in the statement:
$$a_{\tau'}' \leq q(q^{\frac{s+a}{2}}-q^a-1)-1$$
is an equivalent condition to the Hermitian self-orthogonality of the code $\mathcal{C}$.
\end{proof}

\begin{remark}
\label{REM}
{\rm
Propositions \ref{Case1}, \ref{Case2}, \ref{Case3} and \ref{Case4} consider large families of narrow-sense BCH codes and characterize  when codes in these families are Hermitian self-orthogonal. The literature contains some results on Hermitian dual containing BCH codes (regarded as cyclic codes) as we mentioned at the beginning of this section. Most of these results consider primitive  or narrow-sense BCH codes with a small range of lengths.  Our BCH codes and those mentioned in the literature give rise to comparable QECCs. Next we compare our parameters with those in \cite{Aly} and \cite{LLGW2019}, articles that consider a wider range of lengths than others in the literature.

Aly et al. in \cite{Aly} give bounds on $a'_{\tau'}$ that provide QECCs with parameters comparable to ours. Indeed, Theorem 14 in \cite{Aly} states that, if $s$ is even, and
\[
a'_{\tau'} \leq \left\lfloor \frac{n_1}{q^{2s}-1} \left( q^{s+1} - q^2 +1 \right) \right\rfloor -1,
\]
then one gets a QECC with the same parameters as those given by a Hermitian self-orthogonal code BCH$^{n_1} (1, \tau')$. Analogously, if $s$ is odd, the bound is $a'_{\tau'} \leq \left\lfloor \frac{n_1}{q^{s}+1} \right\rfloor -1$. Making computations,  we get (in this case of BCH codes) slight improvements (one unit in the bound for $a'_{\tau'}$ in most situations) in our Case 1 studied in Proposition \ref{Case1}. Concerning Case 2 (Proposition \ref{Case2}), we have $\left\lfloor \frac{n_1}{q^{s}+1} \right\rfloor -1 = n_2 - 1$, which is our bound, so we do not get improvement.

However, when $s$ is odd and we are in Cases 3 or 4 (Propositions \ref{Case3} or \ref{Case4}), one can see that our bound is better than that in \cite{Aly}. Let us show it.

Assume we are in Case 3, where $n_1= (q^s -1) (q^a +1)$, $s > a \neq 0$, $\frac{s+a}{2}$ and $s$ odd. Notice that this last condition implies that $a$ is also odd. Let us prove that the bound in \cite{Aly} is always less than ours. For a start,
\[
\left\lfloor \frac{n_1}{q^{s}+1} \right\rfloor - 1 = \left\lfloor \frac{(q^{s}-1)(q^{a}+1)}{q^{s}+1} \right\rfloor -1  \leq q^{a}-1.
\]
Then, considering our bound given in Proposition \ref{Case3},
\[
q^{a}-1 < q^{\frac{s+a}{2}} - q^{\frac{s-a}{2}} -2 \mbox{ iff }  q^{a} +1 < q^{\frac{s-a}{2}} \left( q^a - 1\right) \mbox{ iff } \;\;  1+ \frac{2}{q^a -1} < q^{\frac{s-a}{2}},
\]
which is always true under our conditions. Notice that the limit case is $a=1$, $s=5$ and $q=2$.

Finally, suppose that we are in Case 4 where $n_1 = (q^s -1) (q^a +1)$, $s > a \neq 0$, $\frac{s+a}{2}$ even and $s$ odd. Again, it implies that $a$ is also odd. Considering the bound given in Proposition \ref{Case4}, one gets
\[
q^{a}-1 < q \left( q^{\frac{s+a}{2}} - q^a -1 \right) \;\; \mbox{ iff } \;\; q^{a}-1 < q q^a \left(q^{\frac{s-a}{2}} -1\right) -q
\]
\[
\mbox{ iff } \; \; 1- \frac{1}{q^a} < q  \left(q^{\frac{s-a}{2}} -1\right) - \frac{1}{q^{a-1}} \;\;\mbox{ iff } \;\; 1- \frac{1}{q^a} + \frac{1}{q^{a-1}} < q  \left(q^{\frac{s-a}{2}} -1\right),
\]
and, again, it holds under our conditions.
%This last inequality is true except in the case $q=2, s=3$ and $a=1$. Taking into account that $q^{a}+1$ is larger than the bound in \cite{Aly}, it is easy to check that, in this case ($q=2, s=3$ and $a=1$), the actual bound in \cite{Aly} and our bound in Proposition \ref{Case3} coincide.

To conclude this remark, it is worthwhile to indicate that \cite{LLGW2019} gives a sharp bound for Hermitian dual containment of a family of BCH codes (regarded as cyclic codes) that gives rise to QECCs intersecting ours. Theorems 3 and 4 in \cite{LLGW2019} assume that $n_1 = \frac{(q^{s}+1)(q^{s}-1)}{b}$, where $b$ divides $q^{s}+1$ and $ 3 \leq b \leq 2(q^2-q+1)$. Our Propositions \ref{Case3} and \ref{Case4} provide some of the above cases and our bounds coincide with those in \cite{LLGW2019}. However, we supply QEECs which are not contemplated in \cite{LLGW2019}. Let us see an example. Set $q=3$, $s=5$ and $a=1$, then $n_1= (3^5-1)(3^1+1)=968$ and clearly we are in our Case 4, but $n_1 = (3^5-1) \frac{244}{61}$, that means $b=61$ and since $61 \nleq 2(9-3 +1)$, it is not under the conditions in \cite{LLGW2019}.}

\end{remark}

\section{New quantum stabilizer codes coming from homothetic-BCH codes}
\label{Sect4}

\subsection{The main result}
\label{Sect41}
Consider a homothetic-BCH code $\Su_{\Delta_1(\tau)}^P$ as introduced in Definition \ref{definit} with respect to a set $\Delta_1(\tau) = \Delta^{U(q^{2s}-1)} (1,\tau)= \Lambda_{a_1} \cup \cdots \cup \Lambda_{a_\tau}$ as given in (\ref{eldelta}). Set $\Delta'_1(\tau')=\Lambda'_{a_1'} \cup \cdots \cup \Lambda'_{a'_{\tau'}}$ the corresponding set introduced in (\ref{eldeltap}).
Gathering Theorems \ref{bounds} and \ref{self} and Propositions \ref{Case1}, \ref{Case2}, \ref{Case3} and \ref{Case4}, keeping the notation as in those results and applying Theorem \ref{col:stabherm}, we conclude the main result in this paper, Theorem \ref{mainth}. Next we state it, and then we give a number of examples that show it allows us to get quantum codes with better parameters than those appearing in the literature.

\begin{thm}\label{mainth}
    Let $q$ be a prime $p$ power and consider $s\geq 2$ a positive integer. Let $n_1$ be a positive integer that divides $q^{2s}-1$ as we are going to introduce, and define
    $$L:=\begin{cases}
			qn_2- \min\left\{ \floor*{\frac{(q-1)n_2}{q^{s-1}+1}},  \floor*{\frac{(q-1)n_2-1}{q^{s-1}}} \right\}-1, & \text{if } n_1=(q^s+1)n_2 \text{, } n_2\mid q^s-1\text{, } s \text{ even,}\\
			n_2-1 & \text{if } n_1=(q^s+1)n_2 \text{, } n_2\mid q^s-1\text{, } s \text{ odd,}\\
            q^{\frac{s+a}{2}}+q^{\frac{s-a}{2}}-2 & \text{if } n_1=(q^s-1)(q^a+1) \text{, } s> a\neq 0 \text{, } \frac{s+a}{2} \text{ odd,}\\
            2q^{\frac{s}{2 }}-3 & \text{if } n_1=2(q^s-1) \text{, } s> 0 \text{, } \frac{s}{2} \text{ odd,}\\
            q(q^{\frac{s+a}{2}}-q^a-1)-1 & \text{if } n_1=(q^s-1)(q^a+1) \text{, } s> a \text{, } \frac{s+a}{2} \text{ even,}\\ & \text{(except when $q=2$ and $a=s-2$).}
	\end{cases}
    $$
Set $\lambda$ another positive integer such that $\lambda\leq \frac{q^{2s}-1}{n_1}$ and $\lambda n_1 \nmid q^{2s}-1$.

Assume that $a'_{\tau'} \leq L$. Then, we have that $\Su_{\Delta_1(\tau)}^P$ is a $[\lambda n_1,\,\leq \sum_{\ell=1}^\tau \# \Lambda_{a_\ell}]_{q^2}$ Hermitian self-orthogonal code whose Hermitian dual has minimum distance at least $a_{\tau+1}$ and, therefore,  there exists a
    $$[[\lambda n_1, \geq \lambda n_1-2\sum_{\ell=1}^\tau \# \Lambda_{a_\ell},\,\geq a_{\tau+1}]]_q$$
quantum stabilizer code.

Additionally assume that $a'_{\tau'} \leq L$ and $p\mid \lambda$. Then, $\Su_{\Delta_0(\tau)}^P$ is a $[\lambda n_1,\,\leq \sum_{\ell=0}^\tau \# \Lambda_{a_\ell}]_{q^2}$ Hermitian self-orthogonal code whose Hermitian dual has minimum distance at least $a_{\tau+1}+1$ and therefore there exists a
$$[[\lambda n_1, \geq \lambda n_1-2\sum_{\ell=0}^\tau \# \Lambda_{a_\ell},\,\geq a_{\tau+1}+1]]_q$$
quantum stabilizer code.
\end{thm}

\subsection{Examples}
\label{Sect42}

In Subsections \ref{421} and \ref{422}, we provide examples obtained by applying the above theorem which are records according to \cite{codetables}. Notice that, although our theorem supplies a lower bound for the dimension of the stabilizer code, we can easily compute the actual one with the computational algebra system MAGMA \cite{Magma}. Finally, in Subsection \ref{423}, by considering a BCH code as introduced in Definition \ref{AA}, we give the parameters of two QECCs which are also records.

\subsubsection{Quantum codes over $\mathbb{F}_2$} \label{421}
%\hfill\break

With the above notation, set $q=2$, $s=5$, $a=1$ and $\lambda=2$. Then, we are in the case where $n_1=(q^s-1)(q^a+1)$, $\frac{s+a}{2}$ odd, $n_1=31 \cdot 3$ and $n=186$.
%Our bound says that we can achieve stabilizer codes of minimum distance $d\ge q^\frac{s+a}{2}+q^\frac{s-a}{2}-1=11$.

%So we may take cyclotomic cosets with representative less than or equal to 10.

Consider the code $\Co=\Su_\Delta^P$, where $P$ is the zero set of the ideal
$$I=\gen*{(X^{n_1}-1)(X^{n_1}-\gamma^{n_1})}_{\mathbb{F}_{2^{10}}[X]},$$
$\gamma$ a primitive element of $\mathbb{F}_{2^{10}}$, and $\Delta$ consists of the union of the following cyclotomic cosets of $\{0, 1, \ldots, 2^{10}-1\}$ with respect to $2^2$:
\begin{multline*}
\Lambda_1 =\{ 1,4,16,64,256 \}, \Lambda_2 =\{2,8,32,128,512  \}, \Lambda_3 =\{3,12,48,192,768  \},\\ \Lambda_5 =\{5,20,80,257,320  \},\Lambda_6 =\{6, 24, 96,384,513  \}, \Lambda_7 =\{ 7,28,112,448,769 \}.
\end{multline*}
Now, the set $\Delta'$ obtained by reducing the elements in $\Delta$ modulo $n_1$ is the union of the following cyclotomic cosets in $E(n_1) =\{0,1, \ldots, n_1-1\}$ with respect to $2^2$:
\begin{multline*}
\Lambda'_1 =\{ 1,4,16,64,70 \}, \Lambda'_2 =\{ 2,8,32,35,47  \}, \Lambda'_3 =\{3,6, 12, 24, 48  \},\\ \Lambda'_5 =\{5,20,41, 71,80 \}, \Lambda_7 =\{ 7, 19, 25,28,76 \}.
\end{multline*}

Then, by Theorem \ref{mainth}, since $L=10$ and $a'_{\tau'} =7 < L$, $\Co$ is a $4$-ary Hermitian self-orthogonal code giving rise to a new quantum stabilizer  code $\mathcal{Q}$ with parameters $[[186,126,\geq 9]]_2$. Moreover, by \cite[Lemma 69]{Ketkar}, one gets new quantum codes with parameters
$[[187,126,\geq 9]]_2$, $[[188,126,\geq 9]]_2$ and $[[189,126,\geq 9]]_2$. According to \cite{codetables}, up to know, there was no construction of codes with such parameters.
%$\mathcal{Q}_2=ExtendCode(\mathcal{Q}_1)=[[187,126,\geq 9]]_2 $, $\mathcal{Q}_3=ExtendCode(\mathcal{Q}_2)=[[188,126,\geq 9]]_2 $ $\mathcal{Q}_4=ExtendCode(\mathcal{Q}_3)=[[189,126,\geq 9]]_2$.

\subsubsection{Quantum codes over $\mathbb{F}_5$} \label{422}
%\hfill\break

Set $q=5$, $s=2$ and $\lambda=2$, which fits in the case $n_1=2(q^s-1)$, $s>0$, $\frac{s}{2}$ odd. Thus, $n_1=48$ and $n=96$.
%Our bounds says that we can achieve codes of minimum distance $d\ge 2q^\frac{s}{2}-2=8$. So we may take cyclotomic cosets with representative less than or equal to 7.

Take the code $\Co_1=\Su_\Delta^P$, where $P$ is the zero set of the ideal
$$I=\gen*{(X^{n_1}-1)(X^{n_1}-\gamma^{n_1})}_{\mathbb{F}_{5^{4}}[X]},$$
$\gamma$ a primitive element of $\mathbb{F}_{5^{4}}$, and $\Delta$ consists of the union of the following cyclotomic cosets of $\{0, 1, \ldots, 5^4 -1\}$ with respect to $5^2$:
\begin{multline}
\label{ejf5}
\Lambda_1 =\{1,25\}, \Lambda_2 =\{2, 50\}, \Lambda_3 =\{3, 75\}, \Lambda_4 =\{4, 100\},\\ \Lambda_5 =\{5,125\}, \Lambda_6 =\{6, 150\}, \Lambda_7 =\{7, 175\}.
\end{multline}
Computing the set $\Delta'$ obtained by reducing the elements in $\Delta$ modulo $n_1$, Theorem \ref{mainth} shows that $\Co_1$ is a $5^2$-ary Hermitian self-orthogonal code giving rise to a new quantum code $\mathcal{Q}_1$ with parameters $[[96,68,\geq 8]]_5$, note that $L=7$ in this case. A new quantum code with parameters $[[97,68,\geq 8]]_5$ can be deduced from \cite[Lemma 69]{Ketkar}.
%By using the length extension propagation rule for quantum codes, we can extend $\mathcal{Q}_1$ obtaining also a new one:

Now, the code $\Co_2 =\Su_\Delta^P$, where $P$ is as we have just described but $\Delta$ consists of the union $\cup_{i=1}^{6} \Lambda_{i}$, $\Lambda_{i}$ being as in (\ref{ejf5}), is a $5^2$-ary Hermitian self-orthogonal code by Theorem \ref{mainth}. Thus it gives rise to a new quantum stabilizer code $\mathcal{Q}_2$ with parameters $[[96,72,\geq 7]]_5$. Again, by \cite[Lemma 69]{Ketkar}, one gets new quantum codes with parameters $[[97,72,\geq 7]]_5$, $[[98,72,\geq 7]]_5$ and $[[99,72,\geq 7]]_5$, According to \cite{codetables}, up to now, there was no construction for codes as those given in this subsection.
%$Q_4=ExtendCode(Q_3)=[[97,72,7]]_5 $, $Q_5=ExtendCode(Q_4)=[[98,72,7]]_5 $, $Q_6=ExtendCode(Q_5)=[[99,72,7]]_5 $.

\subsubsection{Quantum codes over $\mathbb{F}_8$} \label{423}

Our last examples show that regarding BCH codes as defined in Definition \ref{AA}, we are able to provide new QECCs obtained from Hermitian self-orthogonality.

Pick $q=8$, $s=2$ and $N=91$ and consider the $8^2$-ary code $\mathcal{D}= \mathrm{BCH}^{N}(1,9)$ defined by the set
$\Delta^{U(91)}(1,9)$. It is the union of the following cyclotomic cosets with respect to $8^2$ modulo $N$:
\begin{multline*}\Lambda_1 =\{ 1, 64\}, \Lambda_2 =\{2, 37 \},\Lambda_3 =\{3,10 \},\Lambda_4 =\{4,74 \},\Lambda_5 =\{5,47 \},\\ \Lambda_6 =\{6,20 \},\Lambda_7 =\{7, 84 \},\Lambda_8 =\{8, 57 \},\Lambda_9 =\{ 9, 30\}.\end{multline*} Then, $\mathcal{D}$ is an $8^2$-ary Hermitian  self-orthogonal code giving rise to a new quantum stabilizer code $\mathcal{Q}_{\mathcal{D}}$ with parameters $[[91,55,\geq 11]]_8$. By \cite[Lemma 69]{Ketkar}, there is also a quantum code with parameters $[[92,55,\geq 11]]_8$.

As in the above subsections, both quantum codes are new because there was no construction for quantum codes with the above indicated parameters (see \cite{codetables}).

\bibliographystyle{plain}
\bibliography{bibliocodigos}

\begin{thebibliography}{10}

\bibitem{Natu}
Quantum AI and Google Collaborators.
\newblock Quantum error correction below the surface code threshold.
\newblock {\em Nature}, 2024.

\bibitem{Aly}
S.A. Aly, A.~Klappenecker, and P.~K. Sarvepalli.
\newblock On quantum and classical {BCH} codes.
\newblock {\em IEEE Trans. Inf. Theory}, 53(3):1183--1188, 2007.

\bibitem{Anderson}
S.E. Anderson~et al.
\newblock Relative hulls and quantum codes.
\newblock {\em IEEE Trans. Inform. Theory}, 70:3190--3201, 2024.

\bibitem{Aru}
F.~Arute~et al.
\newblock Quantum supremacy using a programmable superconducting processor.
\newblock {\em Nature}, 574:505--510, 2019.

\bibitem{AK}
A.~Ashikhmin and E.~Knill.
\newblock Non-binary quantum stabilizer codes.
\newblock {\em IEEE Trans. Inf. Theory}, 47:3065--3072, 2001.

\bibitem{BallP}
P.~Ball.
\newblock Physicists in {C}hina challenge {G}oogle's 'quantum advantage'.
\newblock {\em Nature}, 588(380), 2020.

\bibitem{Bier}
J.~Bierbrauer.
\newblock The theory of cyclic codes and a generalization to additive codes.
\newblock {\em Des. Codes Cryptogr.}, 25(2):189--206, 2002.

\bibitem{BE}
J.~Bierbrauer and Y.~Edel.
\newblock Quantum twisted codes.
\newblock {\em J. Comb. Designs}, 8:174--188, 2000.

\bibitem{Magma}
W.~Bosma, J.~Cannon, and C.~Playoust.
\newblock The {M}agma algebra system {I}: {T}he user language.
\newblock {\em J. Symbolic Comput.}, 24(3-4):235--265, 1997.

\bibitem{18kkk}
A.R. Calderbank, E.M. Rains, P.W. Shor, and N.J.A. Sloane.
\newblock Quantum error correction and orthogonal geometry.
\newblock {\em Phys. Rev. Lett.}, 76:405--409, 1997.

\bibitem{Calderbank}
A.R. Calderbank, E.M. Rains, P.W. Shor, and N.J.A. Sloane.
\newblock Quantum error correction via codes over {${\rm GF}(4)$}.
\newblock {\em IEEE Trans. Inf. Theory}, 44(4):1369--1387, 1998.

\bibitem{Cas}
I.~Cascudo.
\newblock On squares of cyclic codes.
\newblock {\em IEEE Trans. Inf. Theory}, 65(2):1034--1047, 2019.

\bibitem{GGHR2017}
C.~Galindo, O.~Geil, F.~Hernando, and D.~Ruano.
\newblock On the distance of stabilizer quantum codes from {$J$}-affine variety
  codes.
\newblock {\em Quantum Inf. Process.}, 16(111), 2017.

\bibitem{GH2015}
C.~Galindo and F.~Hernando.
\newblock Quantum codes from affine variety codes and their subfield subcodes.
\newblock {\em Des. Codes Cryptogr.}, 76:89--100, 2015.

\bibitem{Traza}
C.~Galindo, F.~Hernando, and D.~Ruano.
\newblock Classical and quantum evaluation codes at the trace roots.
\newblock {\em IEEE Trans. Inf. Theory}, 65(4):2593--2602, 2019.

\bibitem{Gottesman}
D.~Gottesman.
\newblock Class of quantum error-correcting codes saturating the quantum
  {H}amming bound.
\newblock {\em Phys. Rev. A}, 54(3):1862--1868, 1996.

\bibitem{codetables}
M.~Grassl.
\newblock Bounds on the minimum distance of linear codes.
\newblock {\em www.codetables.de}, accessed on 22/02/2025.

\bibitem{HZC2015}
X.~Hu, G.~Zhang, and B.~Chen.
\newblock Constructions of new nonbinary quantum codes.
\newblock {\em Int. J. Theor. Phys.}, 54:92--99, 2015.

\bibitem{XingC}
L.~Jin, S.~Ling, J.~Luo, and C.~Xing.
\newblock Application of classical {H}ermitian self-orthogonal {MDS} codes to
  quantum {MDS} codes.
\newblock {\em IEEE Trans. Inform. Theory}, 56(9):4735--4740, 2010.

\bibitem{KZT2013}
X.~Kai, S.~Zhu, and Y.~Tang.
\newblock Quantum negacyclic codes.
\newblock {\em Phys. Rev. A}, 88(012326), 2013.

\bibitem{Ketkar}
A.~Ketkar, A.~Klappenecker, S.~Kumar, and P.~K. Sarvepalli.
\newblock Nonbinary stabilizer codes over finite fields.
\newblock {\em IEEE Trans. Inf. Theory}, 52(11):4892--4914, 2006.

\bibitem{lag2}
G.G. La~Guardia.
\newblock On the construction of nonbinary quantum {BCH} codes.
\newblock {\em IEEE Trans. Inf. Theory}, 60(3):1528--1535, 2014.

\bibitem{L2022}
F.~Li.
\newblock The {H}ermitian dual-containing {LCD BCH} codes and related quantum
  codes.
\newblock {\em Cryptogr. Commun.}, 14:579–596, 2022.

\bibitem{LS2021}
F.~Li and X.~Sun.
\newblock The {H}ermitian dual containing non-primitive {BCH} codes.
\newblock {\em IEEE Commun. Lett.}, 25(2):379 -- 382, 2021.

\bibitem{LLGW2019}
Y.~Liu, R.~Li, G.~Guo, and J.~Wang.
\newblock Some nonprimitive {BCH} codes and related quantum codes.
\newblock {\em IEEE Trans. Inf. Theory}, 65(12):7829 -- 7839, 2019.

\bibitem{LLLM2017}
Y.~Liu, R.~Li, L.~Lv, and Y.~Ma.
\newblock A class of constacyclic {BCH} codes and new quantum codes.
\newblock {\em Quantum Inf. Process.}, 16(66), 2017.

\bibitem{LLLG2017}
Y.~Liu, R.~Li, L.~Lü, and L.~Guo.
\newblock New quantum codes derived from a family of antiprimitive {BCH} codes.
\newblock {\em Int. J. Quantum Inf.}, 15(7):1750052, 2017.

\bibitem{LMFL2013}
Y.~Liu, Y.~Ma, Y.~Feng, and R.~Li.
\newblock New quantum codes constructed from a class of imprimitive {BCH}
  codes.
\newblock {\em Int. J. Quantum Inf.}, 11(1):1350006, 2013.

\bibitem{LiuY}
Y.~Liu~et al.
\newblock Verifying quantum advantage experiments with multiple amplitude
  tensor network contraction.
\newblock {\em Phys. Rev. Lett.}, 132(030601), 2024.

\bibitem{QZ2017}
J.~Qian and L.~Zhang.
\newblock Improved constructions for nonbinary quantum {BCH} codes.
\newblock {\em Int. J. Theor. Phys.}, 56:1355–1363, 2017.

\bibitem{SYW2019}
X.~Shi, Q.~Yue, and Y.~Wu.
\newblock The dual-containing primitive {BCH} codes with the maximum designed
  distance and their applications to quantum codes.
\newblock {\em Des. Codes Cryptogr.}, 87:2165–2183, 2019.

\bibitem{Shor1}
P.W. Shor.
\newblock Algorithms for quantum computation: discrete logarithms and
  factoring.
\newblock In {\em Proc. 35th {A}nn. {S}ymp. {F}ound. {C}omp. {S}c., {\it IEEE
  Comp. Soc. Press}}, pages 124--134, 1994.

\bibitem{23RBC}
P.W. Shor.
\newblock Scheme for reducing decoherence in quantum computer memory.
\newblock {\em Phys. Rev. A}, 52(4):2493--2496, 1995.

\bibitem{XL2021}
L.~Xing and Z.~Li.
\newblock Some new quantum {BCH} codes over finite fields.
\newblock {\em Entropy}, 23(6):712, 2021.

\bibitem{XLGM2016}
G.~Xu, R.~Li, L.~Guo, and Y.~Ma.
\newblock New quantum codes constructed from quaternary {BCH} codes.
\newblock {\em Quantum Inf. Process.}, 15:4099–4116, 2016.

\bibitem{YZKL2017}
J.~Yuan, S.~Zhu, X.~Kai, and P.~Li.
\newblock On the construction of quantum constacyclic codes.
\newblock {\em Des. Codes Cryptogr.}, 85:179–190, 2017.

\bibitem{ZZ2021}
H.~Zhang and S.~Zhu.
\newblock New quantum {BCH} codes of length $n=r(q^2-1)$.
\newblock {\em Int. J. Theor. Phys.}, 60:172--184, 2021.

\bibitem{Zhong}
H.S. Zhong~et al.
\newblock Quantum computational advantage using photons.
\newblock {\em Science}, 370:1460--1463, 2020.

\bibitem{ZSL2018}
S.~Zhu, Z.~Sun, and P.~Li.
\newblock A class of negacyclic {BCH} codes and its application to quantum
  codes.
\newblock {\em Des. Codes Cryptogr.}, 86:2139–2165, 2018.

\end{thebibliography}
	
\end{document}